\documentclass[11pt]{article}%
\usepackage{amsfonts}
\usepackage{color}
\usepackage{amsmath}
\usepackage{fullpage}
\usepackage{amssymb}
\usepackage{graphicx}
\usepackage{hyperref}%
\providecommand{\U}[1]{\protect\rule{.1in}{.1in}}
\newtheorem{theorem}{Theorem}

\newtheorem{corollary}[theorem]{Corollary}

\newtheorem{definition}[theorem]{Definition}

\newtheorem{lemma}[theorem]{Lemma}

\newtheorem{proposition}[theorem]{Proposition}
\newtheorem{remark}[theorem]{Remark}

\newenvironment{proof}[1][Proof]{\noindent\textbf{#1.} }{\ \rule{0.5em}{0.5em}}

\let\originalleft\left
\let\originalright\right
\renewcommand{\left}{\mathopen{}\mathclose\bgroup\originalleft}
\renewcommand{\right}{\aftergroup\egroup\originalright}
\begin{document}

\title{Strong converse for the classical capacity of entanglement-breaking and
Hadamard channels via a sandwiched R\'enyi relative entropy}
\author{Mark M. Wilde\thanks{Hearne Institute for Theoretical Physics,
Department of Physics and Astronomy, Center for Computation and Technology,
Louisiana State University, Baton Rouge, Louisiana 70803, USA}
\and Andreas Winter\thanks{ICREA \& F\'{\i}sica Te\`{o}rica: Informaci\'{o} i Fenomens,
Qu\`{a}ntics, Universitat Aut\`{o}noma de Barcelona, ES-08193 Bellaterra (Barcelona), Spain}
\thanks{School of Mathematics, University of Bristol, Bristol BS8 1TW, United Kingdom}
\and Dong Yang\footnotemark[2] \thanks{Laboratory for Quantum Information, China Jiliang University, Hangzhou, Zhejiang 310018, China}}
\date{\today}
\maketitle

\begin{abstract}
A strong converse theorem for the classical capacity of a quantum channel
states that the probability of correctly decoding a classical message
converges exponentially fast to zero in the limit of many channel uses if the
rate of communication exceeds the classical capacity of the channel. Along
with a corresponding achievability statement for rates below the capacity,
such a strong converse theorem enhances our understanding of the capacity as a
very sharp dividing line between achievable and unachievable rates of
communication. Here, we show that such a strong converse theorem holds for the
classical capacity of all entanglement-breaking channels and all Hadamard
channels (the complementary channels of the former). These results follow by
bounding the success probability in terms of a \textquotedblleft
sandwiched\textquotedblright\ R\'{e}nyi relative entropy, by showing that this
quantity is subadditive for all entanglement-breaking and Hadamard channels,
and by relating this quantity to the Holevo capacity. Prior results regarding
strong converse theorems for particular covariant channels emerge as a special
case of our results.

\end{abstract}

\section{Introduction}

One of the most fundamental tasks in quantum information theory is the
transmission of classical data over many independent uses of a quantum
channel, such that, for a fixed rate of communication, the error probability
of the transmission decreases to zero in the limit of many channel uses. The
maximum rate at which this is possible for a given channel is known as the
classical capacity of the channel. Holevo, Schumacher, and Westmoreland (HSW)
\cite{Hol98,PhysRevA.56.131}\ characterized the classical capacity of a
quantum channel $\mathcal{N}$\ in terms of the following formula:%
\begin{equation}
\chi\left(  \mathcal{N}\right)  \equiv\max_{\left\{  p_{X}\left(  x\right)
,\rho_{x}\right\}  }I\left(  X;B\right)  _{\rho}, \label{eq:HSW-formula}%
\end{equation}
where $\left\{  p_{X}\left(  x\right)  ,\rho_{x}\right\}  $ is an ensemble of
quantum states, $I\left(  X;B\right)  _{\rho}\equiv H\left(  X\right)  _{\rho
}+H\left(  B\right)  _{\rho}-H\left(  XB\right)  _{\rho}$ is the quantum
mutual information, and $H\left(  A\right)  _{\sigma}\equiv-$Tr$\left\{
\sigma\log\sigma\right\}  $ is the von Neumann entropy of a state $\sigma$
defined on system $A$.\footnote{Unless stated otherwise, $\log$ always denotes
the base two logarithm.} In the above formula, the quantum mutual information
$I\left(  X;B\right)  $\ is computed with respect to the following
classical-quantum state:%
\begin{equation}
\rho_{XB}\equiv\sum_{x}p_{X}\left(  x\right)  \left\vert x\right\rangle
\left\langle x\right\vert _{X}\otimes\mathcal{N}_{A\rightarrow B}\left(
\rho_{x}\right)  , \label{eq:cq-state}%
\end{equation}
for some orthonormal basis $\left\{  \left\vert x\right\rangle \right\}  $,
and the notation $\mathcal{N}_{A\rightarrow B}$ indicates that the channel
accepts an input on the system $A$ and outputs to the system $B$.

For certain quantum channels, the HSW\ formula is equal to the classical
capacity of the channel
\cite{PhysRevLett.78.3217,AHW00,K02,K03dep,S02,F05,DHS06,KMNR07}. These
results follow because the Holevo formula was shown to be additive for these
channels, in the sense that the following relation holds for these channels
for any positive integer$~n$:%
\[
\chi\left(  \mathcal{N}^{\otimes n}\right)  =n\,\chi\left(  \mathcal{N}%
\right)  .
\]
However, in general, if one cannot show that the Holevo formula is additive
for a given channel, then our best characterization of the classical capacity
is given by a regularized formula:%
\[
\chi_{\text{reg}}\left(  \mathcal{N}\right)  \equiv\lim_{n\rightarrow\infty
}\frac{1}{n}\chi\left(  \mathcal{N}^{\otimes n}\right)  .
\]
The work of Hastings \cite{H09} suggests that the regularized limit is
necessary unless we are able to find some better characterization of the
classical capacity, other than the above one given by HSW. Also, an important
implication of Hastings' result, which demonstrates a strong separation
between the classical and quantum theories of information, is that using
entangled quantum codewords between multiple channel uses can enhance the
classical capacity of certain quantum channels, whereas it is known that
classically correlated codewords do not \cite{Hol98,PhysRevA.56.131,W99,ON99}.

Given the above results, one worthwhile direction is to refine our
understanding of the classical capacity of channels for which the HSW\ formula
is additive. Indeed, the achievability part of the HSW\ coding theorem states
that as long as the rate of communication is below the classical capacity of
the channel, then there exists a coding scheme such that the error probability
of the scheme decreases exponentially fast to zero. The converse part of the
capacity theorem makes use of the well known Holevo bound \cite{Holevo73}, and
it states that if the rate of communication exceeds the capacity, then the
error probability of any coding scheme is bounded away from zero in the limit
of many channel uses.

Such a converse statement as given above might suggest that there is room for
a trade-off between error probability and communication rate. That is, such a
\textquotedblleft weak\textquotedblright\ converse suggests that it might be
possible for one to increase communication rates by allowing for an increased
error probability. A \textit{strong converse theorem} leaves no such room for
a trade-off---it states that if the rate of communication exceeds the
capacity, then the error probability of any coding scheme converges to one in
the limit of many channel uses. Importantly, a strong converse theorem
establishes the capacity of a channel as a very sharp dividing line between
which communication rates are achievable or unachievable in the limit of many
channel uses.

Strong converse theorems hold for all discrete memoryless classical channels
\cite{W62,A73}. Wolfowitz employed a combinatorial approach based on the
theory of types in order to prove the strong converse theorem \cite{W57,W62}.
Arimoto used R\'{e}nyi entropies to bound the probability of successfully
decoding in any communication scheme (hereafter referred to as
\textquotedblleft success probability\textquotedblright) \cite{A73}, as a
counterpart to Gallager's lower bounds on the success probability in terms of
R\'{e}nyi entropies~\cite{G68}. Both the Wolfowitz and Arimoto approaches
demonstrate that the success probability converges exponentially fast to zero
if the rate of communication exceeds the capacity.\footnote{Note that the
earlier approach of Wolfowitz \cite{W57} does not give such a bound, but his
later approach does \cite{W62}.} Much later, Polyanskiy and Verd\'{u}
generalized the Arimoto approach in a very useful way, by showing how to
obtain a bound on the success probability in terms of any
relative-entropy-like quantity satisfying several natural properties
\cite{PV10}.

Less is known about strong converses for quantum channels. However, Winter
\cite{W99} and Ogawa and Nagaoka \cite{ON99} independently proved a strong
converse theorem for channels with classical inputs and quantum outputs.
For such channels, the HSW\ formula in (\ref{eq:HSW-formula}) is equal to the
classical capacity. The proof of the strong converse in Ref.~\cite{W99} used a
combinatorial approach in the spirit of Wolfowitz. Ogawa and Nagaoka's
proof~\cite{ON99} is in the spirit of Arimoto. Both these proofs or proof
techniques show that the strong converse holds for the Holevo capacity (HSW
formula) when restricting to codes for which messages are encoded as product
states (cf.~\cite{WinterPhD99}).

After this initial work, Koenig and Wehner proved that the strong converse
holds for the classical capacity of particular covariant quantum channels
\cite{KW09}. Their proof is in the spirit of Arimoto---they considered a
Holevo-like quantity derived from the quantum R\'enyi relative entropy and
then showed that this quantity is additive for particular covariant channels.
This reduction of the strong converse question to the additivity of an
information quantity is similar to the approach of Arimoto, but the situation
becomes more interesting for the case of quantum channels since entanglement
between channel uses might lead to the quantity being non-additive.

\section{Summary of results}

In this paper, we prove that a strong converse theorem holds for the classical
capacity of all entanglement-breaking channels \cite{H99,S02,HSR03} and their
complementary channels, so-called Hadamard channels \cite{K03had,KMNR07}.

Entanglement-breaking channels can be modeled as the following process:

\begin{enumerate}
\item The channel performs a quantum measurement on the incoming state.

\item The channel then prepares a particular quantum state at the output
depending on the result of the measurement.
\end{enumerate}

The channels are said to be entanglement-breaking because if one applies a
channel in this class to a share of an entangled state, then the resulting
bipartite state is a separable state, having no entanglement. An important
subclass of the entanglement-breaking channels are quantum measurement
channels, in which only the first step above occurs and the output is
classical. A few authors have studied quantum measurement channels and their
corresponding classical capacities in order to interpret the notion of the
information gain of a quantum measurement \cite{J03,H11,DDS11,OCMB11}%
\ (however, see also Refs.~\cite{Winter04,BHH08,Wilde12,BRW13}\ for different
interpretations of the information gain of a quantum measurement).

Hadamard channels are the complementary channels of entanglement-breaking
ones. That is, the map from the input to the environment of an
entanglement-breaking channel is a Hadamard channel. Such channels are given
the name \textquotedblleft Hadamard\textquotedblright\ because their output is
equal to the Hadamard (also known as Schur), i.e.~entry-wise, multiplication
of a representation of the input density matrix with a positive semi-definite
matrix. Some interesting channels fall into this class: generalized dephasing
channels~\cite{DS05,YHD08}, cloning channels \cite{B11,C11}, and the so-called
Unruh channel \cite{B11,BDHM09,BHP09}. The generalized dephasing channel
represents a natural mechanism for decoherence in physical systems such as
superconducting qubits~\cite{BDKS08}, the cloning channel represents a natural
process that occurs during stimulated emission~\cite{MH82,SWZ00,LSHB02}, and
the Unruh channel arises in relativistic quantum information
theory~\cite{B11,BDHM09,BHP09}, bearing connections to the process of
black-hole stimulated emission~\cite{PhysRevD.14.870}.

Our result thus sharpens our understanding of the classical capacity for these
two classes of channels, as motivated in the introduction. Also, there should
be applications of our strong converse theorem in the setting of the noisy
bounded storage model of cryptography as discussed in Ref.~\cite{KWW12}, but
we do not specifically address this application here. Moreover, this paper
introduces an information quantity, dubbed the ``sandwiched R\'enyi relative
entropy,'' and we prove that it satisfies monotonicity under quantum
operations. This quantity should be of independent interest for study in
quantum information theory. It was independently defined in \cite{MDSFT13}.

We now give a brief sketch of the proof of the strong converse for
entanglement-breaking channels, as a guide for the details given in the rest
of the paper. The proof for Hadamard channels follows some of the same steps,
and it ultimately relies on their relation to entanglement-breaking channels
along with some additional steps.

\begin{enumerate}
\item First, we recall the argument of Sharma and Warsi \cite{SW12}\ (which in
turn is based on Ref.~\cite{PV10}), in which they showed that any
relative-entropy-like quantity that satisfies some natural requirements gives
a bound on the success probability of any coding scheme. Let $\mathcal{D}%
\left(  \rho\|\sigma\right)  $\ denote any generalized divergence that
satisfies monotonicity (data processing). From this generalized divergence,
one can define a Holevo-like quantity for a classical-quantum state of the
form in~(\ref{eq:cq-state}), via%
\begin{equation}
\chi_{\mathcal{D}}\left(  \mathcal{N}\right)  \equiv\max_{\left\{
p_{X}\left(  x\right)  ,\rho_{x}\right\}  }I_{\mathcal{D}}\left(  X;B\right)
, \label{eq:generalized-holevo-quantity}%
\end{equation}
where%
\[
I_{\mathcal{D}}\left(  X;B\right)  \equiv\min_{\sigma_{B}}\mathcal{D}\left(
\rho_{XB}\|\rho_{X}\otimes\sigma_{B}\right)  .
\]
Such a quantity itself satisfies a data processing inequality, which we can
then exploit to obtain a bound on the success probability for any $\left(
n,R,\varepsilon\right)  $ code (a code that uses the channel $n$ times at a
fixed rate $R$ and has an error probability no larger than $\varepsilon$).

\item We then introduce a \textquotedblleft sandwiched\textquotedblright%
\ R\'{e}nyi relative entropy, based on a parameter $\alpha$ and defined for
quantum states$~\rho$ and$~\sigma$ as%
\begin{equation}
\widetilde{D}_{\alpha}\left(  \rho\|\sigma\right)  \equiv\frac{1}{\alpha
-1}\log\text{Tr}\left\{  \left(  \sigma^{\frac{1-\alpha}{2\alpha}} \,\rho\,
\sigma^{\frac{1-\alpha}{2\alpha}}\right)  ^{\alpha}\right\}  .
\label{eq:alt-renyi-rel-ent}%
\end{equation}
(See also Ref.~\cite{MDSFT13}). This definition of the R\'{e}nyi relative
entropy is different from the traditional one employed in quantum information
theory \cite{P86} (see Refs.~\cite{KW09,MH11} for applications of this
quantity). Recall that the R\'{e}nyi relative entropy is defined as
\cite{P86}
\[
D_{\alpha}\left(  \rho\|\sigma\right)  \equiv\frac{1}{\alpha-1}\log
\text{Tr}\left\{  \rho^{\alpha} \sigma^{1-\alpha}\right\}  .
\]
However, it follows from the Lieb-Thirring trace inequality \cite{LT76} that
$\widetilde{D}_{\alpha}\left(  \rho\|\sigma\right)  \leq D_{\alpha}\left(
\rho\|\sigma\right)  $ for all $\alpha>1$. Also, one can easily see that the
two quantities are equal to each other whenever $\rho$ and $\sigma$ commute
(when the states are effectively classical).\newline\newline We prove that
$\widetilde{D}_{\alpha}\left(  \rho\|\sigma\right)  $ is monotone under
quantum operations for all $\alpha\in(1,2]$ and that it reduces to the von
Neumann relative entropy in the limit as $\alpha\rightarrow1$. These
properties establish $\widetilde{D}_{\alpha}\left(  \rho\|\sigma\right)  $ as
a relevant information quantity to consider in quantum information theory. In
particular, it will be useful for us in establishing the strong converse for
entanglement-breaking and Hadamard channels. We then define a Holevo-like
quantity $\widetilde{\chi}_{\alpha}\left(  \mathcal{N}\right)  $ via the
recipe given in (\ref{eq:generalized-holevo-quantity}).

\item Combining the above two results, we establish the following upper bound
on the success probability of any rate $R$ classical communication scheme that
uses a channel $n$ times:%
\[
p_{\text{succ}}\leq2^{-n\left(  \frac{\alpha-1}{\alpha}\right)  \left(
R-\frac{1}{n}\widetilde{\chi}_{\alpha}\left(  \mathcal{N}^{\otimes n}\right)
\right)  }.
\]
One can realize by inspecting the above formula that subadditivity of
$\widetilde{\chi}_{\alpha}$ would be helpful in proving the strong converse,
i.e., if the following holds%
\begin{equation}
\widetilde{\chi}_{\alpha}\left(  \mathcal{N}^{\otimes n}\right)  \leq
n\widetilde{\chi}_{\alpha}\left(  \mathcal{N}\right)  .
\label{eq:additivity-for-n}%
\end{equation}

\item Our next step is to prove that the Holevo-like quantity $\widetilde
{\chi}_{\alpha}$\ is equal to an \textquotedblleft$\alpha$-information
radius\textquotedblright\ \cite{S69,C95,MH11}:%
\begin{equation}
\widetilde{\chi}_{\alpha}\left(  \mathcal{N}\right)  =\widetilde{K}_{\alpha
}\left(  \mathcal{N}\right)  \equiv\min_{\sigma}\max_{\rho}\widetilde
{D}_{\alpha}\left(  \mathcal{N}\left(  \rho\right)  \|\sigma\right)  .
\label{eq:holevo-radius-identity}%
\end{equation}
Proving this identity builds upon prior work in Refs.~\cite{SW01,KW09}.

\item At this point, we exploit two observations. First, conjugating a
completely positive entanglement-breaking map by a positive operator does not
take it out of this class---i.e., if $\mathcal{M}_{\operatorname{EB}}$ is a
completely positive entanglement-breaking map, then so is $\mathcal{X}%
\circ\mathcal{M}_{\operatorname{EB}}$ for any positive operator $X$, where the
action of $\mathcal{X}\circ\mathcal{M}_{\operatorname{EB}}$ on a density
operator $\rho$\ is defined by $X\mathcal{M}_{\operatorname{EB}}\left(
\rho\right)  X$. Furthermore, if $\mathcal{M}$ is an arbitrary completely
positive map, then $\mathcal{X}\circ\mathcal{M}$ for any positive $X$ is
completely positive as well. Also, it is possible to interpret the $\alpha
$-information radius $\widetilde{K}_{\alpha}\left(  \mathcal{N}\right)  $ in
terms of a \textquotedblleft sandwiched\textquotedblright\ $\alpha$-norm,
defined as%
\[
\left\Vert A\right\Vert _{\alpha,X}\equiv\left\Vert X^{1/2} A X^{1/2}%
\right\Vert _{\alpha},
\]
for any positive operator $X$ and where%
\[
\left\Vert B\right\Vert _{\alpha}\equiv\ \text{Tr}\{ ( \sqrt{B^{\dag}B})
^{\alpha}\} ^{1/\alpha}.
\]
With these definitions and that in (\ref{eq:alt-renyi-rel-ent}), one can see
that%
\[
\widetilde{K}_{\alpha}\left(  \mathcal{N}\right)  \equiv\min_{\sigma}%
\max_{\rho}\frac{\alpha}{\alpha-1}\log\left\Vert \mathcal{N}\left(
\rho\right)  \right\Vert _{\alpha,\sigma^{\frac{1-\alpha}{\alpha}}}.
\]
King proved that the maximum output $\alpha$-norm of an entanglement-breaking
channel and any other channel is multiplicative \cite{K03} for $\alpha\geq1$,
and Holevo observed that King's proof extends more generally to hold for a
completely positive entanglement-breaking map and any other completely
positive map \cite{H06}. The following inequality then immediately results
from these observations%
\[
\widetilde{K}_{\alpha}(\mathcal{N}_{\operatorname{EB}}\otimes\mathcal{N}%
)\leq\widetilde{K}_{\alpha}(\mathcal{N}_{\operatorname{EB}})+\widetilde
{K}_{\alpha}(\mathcal{N}),
\]
for $\mathcal{N}_{\operatorname{EB}}$ an entanglement-breaking channel and
$\mathcal{N}$ any other channel. With the identity in
(\ref{eq:holevo-radius-identity}), it follows that%
\[
\widetilde{\chi}_{\alpha}(\mathcal{N}_{\operatorname{EB}}\otimes
\mathcal{N})\leq\widetilde{\chi}_{\alpha}(\mathcal{N}_{\operatorname{EB}%
})+\widetilde{\chi}_{\alpha}(\mathcal{N}),
\]
and we can deduce the subadditivity relation in (\ref{eq:additivity-for-n})
for entanglement-breaking channels by an inductive argument.

\item The bound on the success probability for any coding scheme of rate $R$
when using an entanglement-breaking channel then becomes as follows:%
\[
p_{\text{succ}}\leq2^{-n\left(  \frac{\alpha-1}{\alpha}\right)  \left(
R-\widetilde{\chi}_{\alpha}\left(  \mathcal{N}_{\operatorname{EB}}\right)
\right)  }.
\]
Finally, by a standard argument \cite{ON99,SW12}, we can choose $\varepsilon
>0$ such that $\widetilde{\chi}_{\alpha}(\mathcal{N}_{\operatorname{EB}}%
)<\chi(\mathcal{N}_{\operatorname{EB}})+\varepsilon$ for all $\alpha\geq1$ in
some neighborhood of 1, so that the success probability decays exponentially
fast to zero with $n$ if $R>\chi(\mathcal{N}_{\operatorname{EB}})$. The strong
converse theorem for all entanglement-breaking channels then follows.
\end{enumerate}

The next section reviews some preliminary material, and the rest of the paper
proceeds in the order above, giving detailed proofs for each step. After this,
we provide a proof of the strong converse for the classical capacity of
Hadamard channels. We then conclude with a brief summary and a pointer to
concurrent work in Refs.~\cite{L13,MDSFT13,FL13,B13monotone}.

\section{Preliminaries}

\textbf{Operators, norms, states, maps, and channels.} Let $\mathcal{B}\left(
\mathcal{H}\right)  $ denote the algebra of bounded linear operators acting on
a Hilbert space $\mathcal{H}$. We restrict ourselves to finite-dimensional
Hilbert spaces throughout this paper. The $\alpha$-norm of an operator $X$ is
defined as%
\[
\left\Vert X\right\Vert _{\alpha}\equiv\text{Tr}\{ ( \sqrt{X^{\dag}X})
^{\alpha}\} ^{1/\alpha}.
\]
Let $\mathcal{B}\left(  \mathcal{H}\right)  _{+}$ denote the subset of
positive semidefinite operators (we often simply say that an operator is
\textquotedblleft positive\textquotedblright\ if it is positive
semi-definite). We also write $X\geq0$ if $X\in\mathcal{B}\left(
\mathcal{H}\right)  _{+}$. An\ operator $\rho$ is in the set $\mathcal{S}%
\left(  \mathcal{H}\right)  $\ of density operators if $\rho\in\mathcal{B}%
\left(  \mathcal{H}\right)  _{+}$ and Tr$\left\{  \rho\right\}  =1$. The
tensor product of two Hilbert spaces $\mathcal{H}_{A}$ and $\mathcal{H}_{B}$
is denoted by $\mathcal{H}_{A}\otimes\mathcal{H}_{B}$.\ Given a multipartite
density operator $\rho_{AB}\in\mathcal{H}_{A}\otimes\mathcal{H}_{B}$, we
unambiguously write $\rho_{A}=\ $Tr$_{B}\left\{  \rho_{AB}\right\}  $ for the
reduced density operator on system $A$. A linear map $\mathcal{N}%
_{A\rightarrow B}:\mathcal{B}\left(  \mathcal{H}_{A}\right)  \rightarrow
\mathcal{B}\left(  \mathcal{H}_{B}\right)  $\ is positive if $\mathcal{N}%
_{A\rightarrow B}\left(  \sigma_{A}\right)  \in\mathcal{B}\left(
\mathcal{H}_{B}\right)  _{+}$ whenever $\sigma_{A}\in\mathcal{B}\left(
\mathcal{H}_{A}\right)  _{+}$. Let id$_{A}$ denote the identity map acting on
a system $A$. A linear map $\mathcal{N}_{A\rightarrow B}$ is completely
positive if the map id$_{R}\otimes\mathcal{N}_{A\rightarrow B}$ is positive
for a reference system $R$ of arbitrary size. A linear map $\mathcal{N}%
_{A\rightarrow B}$ is trace-preserving if Tr$\left\{  \mathcal{N}%
_{A\rightarrow B}\left(  \tau_{A}\right)  \right\}  =\ $Tr$\left\{  \tau
_{A}\right\}  $ for all input operators $\tau_{A}\in\mathcal{B}\left(
\mathcal{H}_{A}\right)  $. If a linear map is completely positive and
trace-preserving, we say that it is a quantum channel or quantum operation. A
positive operator-valued measure (POVM) is a set $\{ \Lambda_{m} \}$ of
operators satisfying $\Lambda_{m} \geq0 \ \forall m$ and $\sum_{m} \Lambda_{m}
= I$.

\textbf{Entanglement-breaking maps.} Any linear map $\mathcal{M}_{A\rightarrow
B}$\ can be written in the following form:%
\begin{equation}
\mathcal{M}_{A\rightarrow B}\left(  X\right)  =\sum_{x}N_{x}\text{Tr}\left\{
M_{x}X\right\}  , \label{eq:general-linear-map-rep}%
\end{equation}
for some sets of operators $\left\{  N_{x}\right\}  $ and $\left\{
M_{x}\right\}  $. If $N_{x},M_{x}\geq0$ for all $x$, then we say that the map
is \textit{entanglement-breaking} \cite{H99,S02,HSR03,H06}, and one can also
verify that it is completely positive as well. The following conditions are
equivalent for an entanglement-breaking map $\mathcal{M}_{\operatorname{EB}}$:

\begin{enumerate}
\item There is a representation of $\mathcal{M}_{\operatorname{EB}}$ of the
form in (\ref{eq:general-linear-map-rep}) such that $N_{x},M_{x}\geq0$ for all
$x$.

\item The map $\mathcal{M}_{\operatorname{EB}}$ is completely positive and has
a Kraus representation with rank-one Kraus operators, so that%
\[
\mathcal{M}_{\operatorname{EB}}\left(  X\right)  =\sum_{y}\left\vert
\varphi_{y}\right\rangle \left\langle \phi_{y}\right\vert X\left\vert \phi
_{y}\right\rangle \left\langle \varphi_{y}\right\vert ,
\]
for some sets of vectors $\left\{  \left\vert \varphi_{y}\right\rangle
\right\}  $ and $\left\{  \left\vert \phi_{y}\right\rangle \right\}  $.

\item For any integer $d\geq1$ and $\rho_{12}\in\mathcal{S}\left(
\mathcal{H}_{1}\otimes\mathcal{H}_{d}\right)  $, where $\mathcal{H}_{d}$ is a
$d$-dimensional Hilbert space,%
\[
(\mathcal{M}_{\operatorname{EB}}\otimes\text{id}_{d})\left(  \rho_{12}\right)
=\sum_{z}F_{z}\otimes G_{z},
\]
where $F_{z},G_{z}\geq0$ for all $z$.
\end{enumerate}

\begin{remark}
\label{rem:ent-closure}An important observation for the work presented here is
that conjugating an entanglement-breaking map $\mathcal{M}_{\operatorname{EB}%
}$ by a positive operator $Y$ does not take it out of the
entanglement-breaking class. For example, by defining the map $\mathcal{Y}%
\left(  X\right)  =YXY$, one can easily see that%
\[
(\mathcal{Y}\otimes\operatorname{id}_{d})(\mathcal{M}_{\operatorname{EB}%
}\otimes\operatorname{id}_{d})\left(  \rho_{12}\right)  =\sum_{z}%
YF_{z}Y\otimes G_{z},
\]
so that $YF_{z}Y,\,G_{z}\geq0$ for all $z$ and thus $\mathcal{Y}%
\circ\mathcal{M}_{\operatorname{EB}}$ is an entanglement-breaking map if
$\mathcal{M}_{\operatorname{EB}}$ is. (One can check that the other equivalent
conditions still hold as well.)
\end{remark}

The above property is the main reason why our proof of the strong converse
follows from King's proof of the multiplicativity of the maximum output
$\alpha$-norm for entanglement-breaking maps \cite{K03,H06}. King's proof in
turn exploits the following Lieb-Thirring trace inequality \cite{LT76} (see
also \cite{Carlen09}), which holds for $B\geq0$, any operator $C$, and for
$\alpha\geq1$:%
\begin{equation}
\text{Tr}\{ ( CBC^{\dag}) ^{\alpha}\} \leq\text{Tr}\{ ( C^{\dag}C) ^{\alpha
}B^{\alpha}\} . \label{eq:lieb-thirring}%
\end{equation}

An entanglement-breaking map $\mathcal{N}_{\operatorname{EB}}$ is an
\textit{entanglement-breaking channel} if it is also trace-preserving. In this
case, the above conditions are specialized, taking on a physical
interpretation, so that

\begin{enumerate}
\item The set $\left\{  M_{x}\right\}  $ satisfies $\sum_{x}M_{x}=I$ and
corresponds to a positive operator-valued measure. Each operator $N_{x}$ is a
density operator.

\item The sets of vectors $\left\{  \left\vert \varphi_{y}\right\rangle
\right\}  $ and $\left\{  \left\vert \phi_{y}\right\rangle \right\}  $ satisfy
the overcompleteness relation:%
\[
\sum_{y}\left\vert \phi_{y}\right\rangle \left\langle \varphi_{y}|\varphi
_{y}\right\rangle \left\langle \phi_{y}\right\vert =I.
\]

\item The output state $(\mathcal{M}_{\operatorname{EB}}\otimes
\operatorname{id}_{d})\left(  \rho_{12}\right)  $ is a separable state (a
convex combination of product states) for any input.\footnote{This property is
the reason why these channels are said to be \textquotedblleft
entanglement-breaking.\textquotedblright}
\end{enumerate}

\textbf{Complementary maps and Hadamard maps.} A completely positive map
$\mathcal{M}_{A\rightarrow B}$ has a Kraus representation, so that its action
on any input operator $X$ is as follows:%
\[
\mathcal{M}_{A\rightarrow B}\left(  X\right)  =\sum_{x}A_{x}XA_{x}^{\dag},
\]
for some set of operators $\left\{  A_{x}\right\}  $. Such a map is a quantum
channel if it is also trace preserving, which is equivalent to the following
condition on the Kraus operators: $\sum_{x}A_{x}^{\dag}A_{x}=I$. We can define
a linear operator $V_{A\rightarrow BE}$ as follows:%
\[
V_{A\rightarrow BE}\equiv\sum_{x}A_{x}\otimes\left\vert x\right\rangle _{E},
\]
for some orthonormal basis $\left\{  \left\vert x\right\rangle \right\}  $ for
an environment system $E$. We recover the original map $\mathcal{M}%
_{A\rightarrow B}\left(  X\right)  $ by acting first with the linear operator
$V_{A\rightarrow BE}$\ on the input and then taking a partial trace over the
environment system $E$:%
\[
\mathcal{M}_{A\rightarrow B}\left(  X\right)  =\text{Tr}_{E}\left\{
V_{A\rightarrow BE}\left(  X\right)  V_{A\rightarrow BE}^{\dag}\right\}  .
\]
The map \textit{complementary} to $\mathcal{M}_{A\rightarrow B}$, denoted by
$\mathcal{M}_{A\rightarrow E}$ or $\mathcal{M}^{c}$, is recovered by instead
taking a partial trace over the output system $B$:%
\[
\mathcal{M}_{A\rightarrow E}\left(  X\right)  =\text{Tr}_{B}\left\{
V_{A\rightarrow BE}\left(  X\right)  V_{A\rightarrow BE}^{\dag}\right\}  .
\]
Such a map is unique up to a change of basis for the environment system $E$.

In the case that $\mathcal{M}_{A\rightarrow B}$ is a channel, we say that the
linear operator $V_{A\rightarrow BE}$ is a Stinespring dilation of the channel
$\mathcal{M}_{A\rightarrow B}$ \cite{S55} and one can see that it acts as an
isometry. We also say that the map $\mathcal{M}_{A\rightarrow E}$ as defined
above is the \textit{channel} complementary to $\mathcal{M}_{A\rightarrow B}$
if $\mathcal{M}_{A\rightarrow B}$ is a channel.

Finally, we say that a map (channel) is Hadamard if it is complementary to an
entanglement-breaking map (channel) \cite{KMNR07,H06}.

\section{Bounding the success probability with a generalized divergence}

\label{sec:sharma}For convenience of the reader, in this section we now review
the Sharma-Warsi argument that bounds the success probability for any rate $R$
classical communication scheme in terms of a generalized divergence
\cite{SW12}. This argument in turn is based on the classical argument in
Ref.~\cite{PV10}. We include this review for completeness.

A generalized divergence $\mathcal{D}\left(  \rho\|\sigma\right)  $\ is a
mapping from two quantum states $\rho$ and $\sigma$ to an extended real
number.\footnote{An extended real number can be finite or infinite.}
Intuitively, it should be some measure of distinguishability. A generalized
divergence is useful for us if it is monotone under a quantum operation
$\mathcal{N}$, in the sense that%
\[
\mathcal{D}\left(  \rho\|\sigma\right)  \geq\mathcal{D}\left(  \mathcal{N}%
\left(  \rho\right)  \|\mathcal{N}\left(  \sigma\right)  \right)  .
\]
Intuitively, one should not be able to increase the distinguishability of
$\rho$ and $\sigma$ by processing with a noisy quantum operation $\mathcal{N}$.

From the above monotonicity property, we can conclude that $\mathcal{D}\left(
\rho\|\sigma\right)  $ is invariant under tensoring with another quantum state
$\tau$, in the sense that%
\begin{equation}
\mathcal{D}\left(  \rho\otimes\tau\|\sigma\otimes\tau\right)  =\mathcal{D}%
\left(  \rho\|\sigma\right)  . \label{eq:tensoring-with-another-state}%
\end{equation}
This is because tensoring with another system is a CPTP\ map, so that
$\mathcal{D}\left(  \rho\|\sigma\right)  \geq\mathcal{D}\left(  \rho
\otimes\tau\|\sigma\otimes\tau\right)  $, while the partial trace is a
CPTP\ map as well, so that $\mathcal{D}\left(  \rho\otimes\tau\|\sigma
\otimes\tau\right)  \geq\mathcal{D}\left(  \rho\|\sigma\right)  $. The
interpretation of (\ref{eq:tensoring-with-another-state}) is that the
distinguishability of $\rho$ and $\sigma$ should be the same if we append an
additional quantum system in the state $\tau$.

We can also conclude that it is invariant under the application of a unitary
$U$, in the sense that%
\[
\mathcal{D}\left(  \rho\|\sigma\right)  =\mathcal{D}( U\rho U^{\dag}\|U\sigma
U^{\dag}) .
\]
This follows because the maps $U\left(  \cdot\right)  U^{\dag}$ and $U^{\dag
}\left(  \cdot\right)  U$ are CPTP, so that%
\begin{align*}
\mathcal{D}\left(  \rho\|\sigma\right)   &  \geq\mathcal{D}( U\rho U^{\dag
}\|U\sigma U^{\dag}) ,\\
\mathcal{D}( U\rho U^{\dag}\|U\sigma U^{\dag})  &  \geq\mathcal{D}( U^{\dag
}U\rho U^{\dag}U\|U^{\dag}U\sigma U^{\dag}U) =\mathcal{D}\left(  \rho
\|\sigma\right)  .
\end{align*}
From this, we can conclude that the divergence reduces to a classical
divergence (independent of any orthonormal basis) for the case of commuting,
qubit states. Let%
\begin{align*}
\rho_{p}  &  \equiv p\left\vert 0\right\rangle \left\langle 0\right\vert
+\left(  1-p\right)  \left\vert 1\right\rangle \left\langle 1\right\vert ,\\
\rho_{q}  &  \equiv q\left\vert 0\right\rangle \left\langle 0\right\vert
+\left(  1-q\right)  \left\vert 1\right\rangle \left\langle 1\right\vert ,
\end{align*}
for $0\leq p,q\leq1$ and some orthonormal basis $\left\{  \left\vert
0\right\rangle ,\left\vert 1\right\rangle \right\}  $. Let%
\[
\delta\left(  p\|q\right)  \equiv\mathcal{D}\left(  \rho_{p}\|\rho_{q}\right)
.
\]
It follows that $\delta\left(  p\|q\right)  $ is independent of the choice of
basis $\left\{  \left\vert 0\right\rangle ,\left\vert 1\right\rangle \right\}
$.

From such a generalized divergence, we can then define a generalized Holevo
information of a channel $\mathcal{N}$ as%
\begin{equation}
\chi_{\mathcal{D}}\left(  \mathcal{N}\right)  \equiv\max_{\left\{
p_{X}\left(  x\right)  ,\rho_{x}\right\}  }I_{\mathcal{D}}\left(  X;B\right)
_{\rho}, \label{eq:generalized-holevo}%
\end{equation}
where the optimization is over ensembles $\left\{  p_{X}\left(  x\right)
,\rho_{x}\right\} $ and%
\begin{align*}
I_{\mathcal{D}}\left(  X;B\right)  _{\rho}  &  \equiv\min_{\sigma_{B}%
}\mathcal{D}\left(  \rho_{XB}\|\rho_{X}\otimes\sigma_{B}\right)  ,\\
\rho_{XB}  &  \equiv\sum_{x}p_{X}\left(  x\right)  \left\vert x\right\rangle
\left\langle x\right\vert \otimes\mathcal{N}\left(  \rho_{x}\right)  ,
\end{align*}
where the optimization is over states $\sigma_{B}$. It is straightforward to
show that the quantity $I_{\mathcal{D}}\left(  X;B\right)  $ obeys a data
processing inequality by exploiting the fact that the generalized divergence
$\mathcal{D}$ does (see Lemma~1 of Ref.~\cite{SW12} for an explicit proof). In
this case, a data processing inequality means that%
\[
I_{\mathcal{D}}\left(  X;B\right)  _{\rho}\geq I_{\mathcal{D}}\left(
X;B^{\prime}\right)  _{\omega},
\]
for $\omega_{XB^{\prime}}\equiv\left(  \text{id}_{X}\otimes\mathcal{E}%
_{B\rightarrow B^{\prime}}\right)  \left(  \rho_{XB}\right)  $, where
$\mathcal{E}_{B\rightarrow B^{\prime}}$ is a CPTP\ map.

\subsection{Converse bound from a generalized divergence}

We now review the converse argument from Refs.~\cite{SW12,PV10}\ that gives a
bound on the success probability for any rate $R$ scheme for classical
communication. Any $\left(  n,R,\varepsilon\right)  $\ protocol for
communication has the following form: A sender chooses a message uniformly at
random from a message set $\mathsf{M}\equiv\left\{  1,\ldots,\left\vert
\mathsf{M}\right\vert \right\}  $, where $\left\vert \mathsf{M}\right\vert
=2^{nR}$ (it suffices for our purposes to suppose that the choice is uniform).
The sender transmits a quantum state $\rho_{m}$ (a quantum codeword) through
$n$ uses of\ the channel$~\mathcal{N}$. The overall state at this point is
described by the following classical-quantum state:%
\[
\rho_{MB^{n}}\equiv\sum_{m}\frac{1}{\left\vert \mathsf{M}\right\vert
}\left\vert m\right\rangle \left\langle m\right\vert _{M}\otimes
\mathcal{N}^{\otimes n}\left(  \rho_{m}\right)  .
\]
The receiver applies a decoding POVM $\left\{  \Lambda_{m}\right\}  $ to the
output of the channel to produce an estimate $\hat{M}$ of message $M$. The
resulting classical-quantum state is as follows:%
\[
\omega_{M\hat{M}}\equiv\sum_{m,m^{\prime}}\frac{1}{\left\vert \mathsf{M}%
\right\vert }\left\vert m\right\rangle \left\langle m\right\vert _{M}%
\otimes\text{Tr}\left\{  \Lambda_{m^{\prime}}\mathcal{N}^{\otimes n}\left(
\rho_{m}\right)  \right\}  \left\vert m^{\prime}\right\rangle \left\langle
m^{\prime}\right\vert _{\hat{M}}.
\]
The error probability of the scheme is $\varepsilon$ if $\Pr\{\hat{M}\neq
M\}\leq\varepsilon$. Also, without loss of generality, we can assume that
$\varepsilon\leq1-2^{-nR}$ (otherwise, the strong converse would already hold
for rates above the capacity since the error probability would obey the bound
$\varepsilon>1-2^{-nR}$). We now show how to establish the following bound for
any communication scheme as discussed above:%
\begin{equation}
\delta\left(  \varepsilon\|1-2^{-nR}\right)  \leq\chi_{\mathcal{D}}\left(
\mathcal{N}^{\otimes n}\right)  . \label{eq:p_succ-bound}%
\end{equation}

Let $\sigma_{B^{n}}$ denote an arbitrary density operator on the $B^{n}$
systems. From the properties of a generalized divergence and the specification
above, we can deduce that%
\begin{align*}
\mathcal{D}\left(  \rho_{MB^{n}}\Vert\rho_{M}\otimes\sigma_{B^{n}}\right)   &
\geq\mathcal{D}\left(  \omega_{M\hat{M}}\Vert\omega_{M}\otimes\tau_{\hat{M}%
}\right) \\
&  \geq\delta(\Pr\{\hat{M}\neq M\}\|1-2^{-nR})\\
&  \geq\delta\left(  \varepsilon\|1-2^{-nR}\right)  .
\end{align*}
The first inequality follows from monotonicity of the generalized divergence
under the decoding map $\sum_{m}$Tr$\left\{  \Lambda_{m}\left(  \cdot\right)
\right\}  \left\vert m\right\rangle \left\langle m\right\vert _{\hat{M}}$.
Also, here, we are letting%
\[
\tau_{\hat{M}}\equiv\sum_{m}\text{Tr}\left\{  \Lambda_{m}\sigma_{B^{n}%
}\right\}  \left\vert m\right\rangle \left\langle m\right\vert _{\hat{M}}.
\]
The second inequality follows from monotonicity of the generalized divergence
under the \textquotedblleft equality test,\textquotedblright\ which is a
classical map testing if the value in $M$ is equal to the value in $\hat{M}$,
i.e., $(M,\hat{M})\rightarrow\delta_{M,\hat{M}}$ (with $\delta_{x,y}$ the
Kronecker delta function). This test produces the distribution $(\Pr\{\hat
{M}\neq M\},\Pr\{\hat{M}=M\})$ when acting on the state $\omega_{M\hat{M}}$
and the distribution $\left(  1-2^{-nR},2^{-nR}\right)  $ when acting on the
product state $\omega_{M}\otimes\tau_{\hat{M}}$. The last inequality follows
from the monotonicity $\delta\left(  p^{\prime}\|q\right)  \leq\delta\left(
p\|q\right)  $ whenever $p\leq p^{\prime}\leq q$ \cite{PV10} (recall that we
have $\Pr\{\hat{M}\neq M\}\leq\varepsilon\leq1-2^{-nR}$). Given that
$\sigma_{B^{n}}$ is an arbitrary density operator, we can recover the tightest
upper bound on $\delta\left(  \varepsilon\|1-2^{-nR}\right)  $ by minimizing
$\mathcal{D}$ with respect to all such $\sigma_{B^{n}}$:%
\[
\delta\left(  \varepsilon\|1-2^{-nR}\right)  \leq\min_{\sigma_{B^{n}}%
}\mathcal{D}\left(  \rho_{MB^{n}}\Vert\rho_{M}\otimes\sigma_{B^{n}}\right)  .
\]
Finally, we can remove the dependence on the particular code by maximizing
over all input ensembles:%
\begin{align*}
\delta\left(  \varepsilon\|1-2^{-nR}\right)   &  \leq\max_{\left\{
p_{X}\left(  x\right)  ,\rho_{x}\right\}  }\min_{\sigma_{B^{n}}}%
\mathcal{D}\left(  \rho_{XB^{n}}\Vert\rho_{X}\otimes\sigma_{B^{n}}\right) \\
&  =\chi_{\mathcal{D}}\left(  \mathcal{N}^{\otimes n}\right)  ,
\end{align*}
where%
\[
\rho_{XB^{n}}\equiv\sum_{x}p_{X}\left(  x\right)  \left\vert x\right\rangle
\left\langle x\right\vert _{X}\otimes\mathcal{N}^{\otimes n}\left(  \rho
_{x}\right)
\]
and the second line follows from the definition of $\chi_{\mathcal{D}}$.

\begin{remark}
In light of the above bound in terms of a generalized divergence, in
hindsight, the approach of Arimoto \cite{A73}\ (and the follow-up work
\cite{ON99,KW09}) appears to be somewhat ad hoc. This becomes amplified in the
case of proving strong converse theorems for quantum channels, where one can
choose from many different divergences that all reduce to the same classical
divergence. In the next section, we define a divergence which gives bounds on
the success probability that are tighter than those from
Refs.~\cite{ON99,KW09}.
\end{remark}

\begin{remark}
If one employs the von Neumann relative entropy as the divergence, then one
arrives at the following weak converse bound:%
\[
R\leq\frac{1}{n\left(  1-\varepsilon\right)  }\left(  \chi\left(
\mathcal{N}^{\otimes n}\right)  +h_{2}\left(  \varepsilon\right)  \right)  ,
\]
where $h_{2}\left(  \varepsilon\right)  \equiv-\varepsilon\log\varepsilon
-\left(  1-\varepsilon\right)  \log\left(  1-\varepsilon\right)  $.
\end{remark}

\section{The sandwiched quantum R\'{e}nyi relative entropy}

We now define a ``sandwiched'' quantum R\'{e}nyi relative entropy and prove
several of its properties that establish its utility as an information
measure. In particular, the sandwiched R\'{e}nyi relative entropy is based on
a parameter $\alpha$, and its most important property is that it is monotone
under quantum operations for $\alpha\in(1,2]$. We define this quantity more
generally on the space of positive operators, since it might find other
applications in quantum information theory.

We begin by defining a quasi-relative entropy, in the spirit of \cite{P86},
and from this, we obtain the sandwiched R\'{e}nyi relative entropy.

\begin{definition}
\label{def:sandwich-rel-ent}The sandwiched quasi-relative entropy
$\widetilde{Q}_{\alpha}\left(  A\|B\right)  $ is defined for every $\alpha
\in(1,\infty)$ and for $A,B\in\mathcal{B}\left(  \mathcal{H}\right)  _{+}$ as%
\[
\widetilde{Q}_{\alpha}\left(  A\|B\right)  \equiv\left\{
\begin{array}
[c]{cc}%
\operatorname{Tr}\left\{  \left(  B^{\frac{1-\alpha}{2\alpha}}AB^{\frac
{1-\alpha}{2\alpha}}\right)  ^{\alpha}\right\}  & \text{if }
\operatorname{supp}\left(  A\right)  \subseteq\operatorname{supp}\left(
B\right) \\
\infty & \text{otherwise}%
\end{array}
\right.  .
\]
The sandwiched R\'{e}nyi relative entropy is defined as%
\[
\widetilde{D}_{\alpha}\left(  A\|B\right)  \equiv\frac{1}{\alpha-1}%
\log\widetilde{Q}_{\alpha}\left(  A\|B\right)  .
\]

\end{definition}

The sandwiched R\'{e}nyi relative entropy $\widetilde{D}_{\alpha}$ was
independently defined in~\cite{T12tutorial,F13,MDSFT13}. One could certainly
define these quantities for all non-negative $\alpha$, but we only define it
for the above range for simplicity since we use it just for $\alpha\in(1,2]$.

One might suspect that there should be a relation between the sandwiched
relative entropy and the traditional one. Recall that the quantum R\'{e}nyi
relative entropy is defined as%
\begin{equation}
D_{\alpha}\left(  A\|B\right)  \equiv\frac{1}{\alpha-1}\log\text{Tr}\left\{
A^{\alpha}B^{1-\alpha}\right\}  . \label{eq:traditional-renyi-entropy}%
\end{equation}
By applying the Lieb-Thirring inequality from (\ref{eq:lieb-thirring}), we see
that the following inequality holds for all $\alpha>1$:%
\begin{equation}
\widetilde{D}_{\alpha}\left(  A\|B\right)  \leq D_{\alpha}\left(  A\|B\right)
. \label{eq:ineq-d-tilde-alpha-d-alpha}%
\end{equation}
This relationship is the main reason why the sandwiched R\'{e}nyi relative
entropy allows us to obtain tighter upper bounds on the success probability of
any rate $R$ classical communication protocol. Furthermore, whenever $A$ and
$B$ commute, both of these entropies are equal and reduce to the classical
R\'{e}nyi relative entropy. That is, suppose that $A=\sum_{x}a_{x}\left\vert
x\right\rangle \left\langle x\right\vert $ and $B=\sum_{x}b_{x}\left\vert
x\right\rangle \left\langle x\right\vert $. Then both quantities are equal to
the classical R\'{e}nyi relative entropy in such a case:%
\[
\widetilde{D}_{\alpha}\left(  A\|B\right)  =D_{\alpha}\left(  A\|B\right)
=\frac{1}{\alpha-1}\log\sum_{x}a_{x}^{\alpha}\,b_{x}^{1-\alpha}.
\]

We now prove four different properties of the sandwiched quasi-relative
entropy $\widetilde{Q}_{\alpha}\left(  A\|B\right)  $: unitary invariance,
multiplicativity under tensor-product operators, invariance under tensoring
with another system, and joint convexity in its arguments. These four
properties taken together then allow us to conclude that $\widetilde
{Q}_{\alpha}\left(  A\|B\right)  $ is monotone under noisy quantum operations.
Monotonicity of $\widetilde{Q}_{\alpha}\left(  A\|B\right)  $\ then implies
that$\ \widetilde{D}_{\alpha}\left(  A\|B\right)  $ is monotone as well.

\begin{theorem}
The sandwiched quasi-relative entropy $\widetilde{Q}_{\alpha}\left(
A\|B\right)  $\ is invariant under all unitaries $U$, multiplicative under
tensor-product operators $A_{1}\otimes A_{2}$ and $B_{1}\otimes B_{2}$, and
invariant under tensoring $A$ and $B$ with another quantum system:%
\begin{align*}
\widetilde{Q}_{\alpha}\left(  UAU^{\dag}\|UBU^{\dag}\right)   &
=\widetilde{Q}_{\alpha}\left(  A\|B\right)  ,\\
\widetilde{Q}_{\alpha}\left(  A_{1}\otimes A_{2}\|B_{1}\otimes B_{2}\right)
&  =\widetilde{Q}_{\alpha}\left(  A_{1}\|B_{1}\right)  \ \widetilde{Q}%
_{\alpha}\left(  A_{2}\|B_{2}\right)  ,\\
\widetilde{Q}_{\alpha}\left(  A\otimes\tau\|B\otimes\tau\right)   &
=\widetilde{Q}_{\alpha}\left(  A\|B\right)  .
\end{align*}
For all $\alpha\in(1,2]$, the sandwiched quasi-relative entropy $\widetilde
{Q}_{\alpha}\left(  A\|B\right)  $\ is jointly convex in its arguments%
\[
\sum_{x}p\left(  x\right)  \widetilde{Q}_{\alpha}\left(  A_{x}\|B_{x}\right)
\geq\widetilde{Q}_{\alpha}\left(  A\|B\right)  .
\]
where $A=\sum_{x}p\left(  x\right)  A_{x}$ and $B=\sum_{x}p\left(  x\right)
B_{x}$.
\end{theorem}

\begin{proof}
We establish unitary invariance by%
\begin{align*}
\widetilde{Q}_{\alpha}\left(  UAU^{\dag}\|UBU^{\dag}\right)   &
=\text{Tr}\left\{  \left(  \left(  UBU^{\dag}\right)  ^{\frac{1-\alpha
}{2\alpha}}\left(  UAU^{\dag}\right)  \left(  UBU^{\dag}\right)
^{\frac{1-\alpha}{2\alpha}}\right)  ^{\alpha}\right\} \\
&  =\text{Tr}\left\{  \left(  UB^{\frac{1-\alpha}{2\alpha}}U^{\dag}\left(
UAU^{\dag}\right)  UB^{\frac{1-\alpha}{2\alpha}}U^{\dag}\right)  ^{\alpha
}\right\} \\
&  =\text{Tr}\left\{  \left(  UB^{\frac{1-\alpha}{2\alpha}}AB^{\frac{1-\alpha
}{2\alpha}}U^{\dag}\right)  ^{\alpha}\right\} \\
&  =\text{Tr}\left\{  U\left(  B^{\frac{1-\alpha}{2\alpha}}AB^{\frac{1-\alpha
}{2\alpha}}\right)  ^{\alpha}U^{\dag}\right\} \\
&  =\widetilde{Q}_{\alpha}\left(  A\|B\right)  .
\end{align*}

Multiplicativity under tensor-product operators follows because%
\begin{align*}
\widetilde{Q}_{\alpha}\left(  A_{1}\otimes A_{2}\|B_{1}\otimes B_{2}\right)
&  =\text{Tr}\left\{  \left(  \left(  B_{1}\otimes B_{2}\right)
^{\frac{1-\alpha}{2\alpha}}\left(  A_{1}\otimes A_{2}\right)  \left(
B_{1}\otimes B_{2}\right)  ^{\frac{1-\alpha}{2\alpha}}\right)  ^{\alpha
}\right\} \\
&  =\text{Tr}\left\{  \left(  \left(  B_{1}^{\frac{1-\alpha}{2\alpha}}\otimes
B_{2}^{\frac{1-\alpha}{2\alpha}}\right)  \left(  A_{1}\otimes A_{2}\right)
\left(  B_{1}^{\frac{1-\alpha}{2\alpha}}\otimes B_{2}^{\frac{1-\alpha}%
{2\alpha}}\right)  \right)  ^{\alpha}\right\} \\
&  =\text{Tr}\left\{  \left(  B_{1}^{\frac{1-\alpha}{2\alpha}}A_{1}%
B_{1}^{\frac{1-\alpha}{2\alpha}}\otimes B_{2}^{\frac{1-\alpha}{2\alpha}}%
A_{2}B_{2}^{\frac{1-\alpha}{2\alpha}}\right)  ^{\alpha}\right\} \\
&  =\text{Tr}\left\{  \left(  B_{1}^{\frac{1-\alpha}{2\alpha}}A_{1}%
B_{1}^{\frac{1-\alpha}{2\alpha}}\right)  ^{\alpha}\otimes\left(  B_{2}%
^{\frac{1-\alpha}{2\alpha}}A_{2}B_{2}^{\frac{1-\alpha}{2\alpha}}\right)
^{\alpha}\right\} \\
&  =\widetilde{Q}_{\alpha}\left(  A_{1}\|B_{1}\right)  \ \widetilde{Q}%
_{\alpha}\left(  A_{2}\|B_{2}\right)  .
\end{align*}
Invariance under tensoring with another system then follows as a special case
of multiplicativity since we assume that Tr$\left\{  \tau\right\}  =1$.

Finally, we prove that this quantity is jointly convex in its arguments
$A=\sum_{x}p\left(  x\right)  A_{x}$ and $B=\sum_{x}p\left(  x\right)  B_{x}$
whenever $\alpha\in(1,2]$:%
\[
\sum_{x}p\left(  x\right)  \widetilde{Q}_{\alpha}\left(  A_{x}\|B_{x}\right)
\geq\widetilde{Q}_{\alpha}\left(  A\|B\right)  .
\]
Taking $\left\vert \gamma\right\rangle =\sum_{i}\left\vert i\right\rangle
\left\vert i\right\rangle $, we can rewrite $\widetilde{Q}_{\alpha}\left(
A\|B\right)  $ as%
\[
\widetilde{Q}_{\alpha}\left(  A\|B\right)  =\text{Tr}\left\{  \left\vert
\gamma\right\rangle \left\langle \gamma\right\vert \sqrt{g\left(  B\right)
}f\left(  g\left(  B\right)  ^{-1/2}h\left(  A\right)  g\left(  B\right)
^{-1/2}\right)  \sqrt{g\left(  B\right)  }\right\}  ,
\]
where%
\begin{align*}
f\left(  x\right)   &  \equiv x^{\alpha},\\
g\left(  B\right)   &  \equiv B^{\frac{\alpha-1}{\alpha}}\otimes\left(
B^{T}\right)  ^{\frac{1}{\alpha}},\\
h\left(  A\right)   &  \equiv A\otimes I.
\end{align*}
The function $f\left(  x\right)  $ is operator convex for $\alpha\in(1,2]$.
Also, $g\left(  B\right)  $ is operator concave for $\alpha\in(1,2]$ because
$\left(  L,R\right)  \longmapsto L^{x}\otimes R^{y}$ is jointly operator
concave on positive operators for $x,y\geq0$ and $x+y\leq1$ (see Corollary 5.5
of \cite{Wolf12}). Also, $h\left(  A\right)  $ is clearly affine. With all of
this, it follows from Theorem~5.14 of \cite{Wolf12} that
\[
\sqrt{g\left(  B\right)  }f\left(  g\left(  B\right)  ^{-1/2}h\left(
A\right)  g\left(  B\right)  ^{-1/2}\right)  \sqrt{g\left(  B\right)  }%
\]
is jointly operator convex. This then implies that the functional
$\widetilde{Q}_{\alpha}\left(  A\|B\right)  $ is jointly convex in its arguments.
\end{proof}

Monotonicity of $\widetilde{Q}_{\alpha}\left(  A\|B\right)  $ then follows by
using the above properties and a standard argument detailed in Theorem~5.16 of
\cite{Wolf12}. Also, by inspecting the definition of $\widetilde{D}_{\alpha
}\left(  A\|B\right)  $, it follows that $\widetilde{D}_{\alpha}\left(
A\|B\right)  $ is monotone given that $\widetilde{Q}_{\alpha}\left(
A\|B\right)  $ is.

For convenience of the reader, this paper's appendix reproduces the statements
of Theorem~5.14, Corollary~5.5, and Theorem~5.16 from \cite{Wolf12}.

\begin{corollary}
[Monotonicity]For all $\alpha\in(1,2]$, the sandwiched quasi-relative entropy
$\widetilde{Q}_{\alpha}$ and the sandwiched R\'{e}nyi relative entropy
$\widetilde{D}_{\alpha}$ are monotone under a quantum operation $\mathcal{N}$:%
\begin{align*}
\widetilde{Q}_{\alpha}\left(  A\|B\right)   &  \geq\widetilde{Q}_{\alpha
}\left(  \mathcal{N}\left(  A\right)  \|\mathcal{N}\left(  B\right)  \right)
,\\
\widetilde{D}_{\alpha}\left(  A\|B\right)   &  \geq\widetilde{D}_{\alpha
}\left(  \mathcal{N}\left(  A\right)  \|\mathcal{N}\left(  B\right)  \right)
.
\end{align*}

\end{corollary}

We note that this corollary generalizes Theorem~21 of Ref.~\cite{DFW13} beyond
$\alpha=2$ (the above proof of joint convexity of $\widetilde{Q}_{\alpha
}\left(  A\|B\right)  $ is in fact a straightforward generalization of the
proof of Theorem~21 in Ref.~\cite{DFW13}).

\begin{corollary}
[Positivity]The sandwiched R\'{e}nyi relative entropy $\widetilde{D}_{\alpha
}\left(  \rho\|\sigma\right)  $ is non-negative for density operators $\rho$
and $\sigma$ and for $\alpha\in(1,2]$.
\end{corollary}

\begin{proof}
Writing a spectral decomposition for $\rho$ as $\rho=\sum_{x}p\left(
x\right)  \left\vert \phi_{x}\right\rangle \left\langle \phi_{x}\right\vert $,
we can apply a \textquotedblleft dephasing\textquotedblright\ or
\textquotedblleft pinching\textquotedblright\ map $\Delta\left(  \cdot\right)
\equiv\sum_{x}\left\vert \phi_{x}\right\rangle \left\langle \phi
_{x}\right\vert \left(  \cdot\right)  \left\vert \phi_{x}\right\rangle
\left\langle \phi_{x}\right\vert $ to both states. From monotonicity, we find
that%
\[
\widetilde{D}_{\alpha}\left(  \rho\|\sigma\right)  \geq\widetilde{D}_{\alpha
}\left(  \Delta\left(  \rho\right)  \|\Delta\left(  \sigma\right)  \right)
\geq0,
\]
where the second inequality follows because the sandwiched R\'{e}nyi relative
entropy reduces to the classical one, which we know is non-negative for
probability distributions.
\end{proof}

\begin{corollary}
[Equality conditions]For density operators $\rho$ and $\sigma$ and $\alpha
\in(1,2]$, the sandwiched R\'{e}nyi relative entropy satisfies $\widetilde
{D}_{\alpha}\left(  \rho\|\sigma\right)  =0$ if and only if $\rho=\sigma$.
\end{corollary}

\begin{proof}
If $\rho=\sigma$, then $\widetilde{D}_{\alpha}\left(  \rho\|\sigma\right)  =0$
simply by inspecting the definition of the sandwiched R\'{e}nyi relative
entropy. Now suppose that $\widetilde{D}_{\alpha}\left(  \rho\|\sigma\right)
=0$. In this case, we can perform an informationally-complete measurement map
on the states $\rho$ and $\sigma$ \cite{P77,B91,RBSC04}. Such a measurement
map has the following form:
\[
\mathcal{M}(\omega)=\sum_{x}\operatorname{Tr}\{M_{x}\omega\}|x\rangle\langle
x|,
\]
for some orthonormal basis $\{|x\rangle\}$ and operators $M_{x}$ such that
$M_{x}\geq0$ for all $x$ and $\sum_{x}M_{x}=I$, and it is informationally
complete in the sense that all the parameters of the density operator $\omega$
are encoded in the distribution $\operatorname{Tr}\{M_{x}\omega\}$ of the
outcomes. From monotonicity and positivity of the sandwiched R\'{e}nyi
relative entropy under quantum operations, it follows that $\widetilde
{D}_{\alpha}(\mathcal{M}(\rho)\|\mathcal{M}(\sigma))=0$. But this R\'{e}nyi
relative entropy is with respect to classical states, and it is known that the
equality conditions for the classical R\'{e}nyi relative entropies are that
$\widetilde{D}_{\alpha}(\mathcal{M}(\rho)\|\mathcal{M}(\sigma))=0$ if and only
if $\operatorname{Tr}\{M_{x}\rho\}=\operatorname{Tr}\{M_{x}\sigma\}$ for all
$x$ \cite{C95}. Since we chose the measurement to be informationally complete,
it follows that $\rho=\sigma$.

An alternate proof of the implication $\widetilde{D}_{\alpha}\left(  \rho
\Vert\sigma\right)  =0\implies\rho=\sigma$, suggested by an anonymous referee,
is as follows. Let $U$ be any unitary and let $\Delta$ be the dephasing or
pinching map given above. Then we have%
\[
0=\widetilde{D}_{\alpha}\left(  \rho\Vert\sigma\right)  =\widetilde{D}%
_{\alpha}\left(  U\rho U^{\dag}\Vert U\sigma U^{\dag}\right)  \geq
\widetilde{D}_{\alpha}\left(  \Delta\left(  U\rho U^{\dag}\right)  \Vert
\Delta\left(  U\sigma U^{\dag}\right)  \right)  =0.
\]
By the classical conditions for equality, it follows that $\Delta\left(
U\left(  \rho-\sigma\right)  U^{\dag}\right)  =0$ for any unitary $U$. But
then it immediately follows that Tr$\left\{  B\left(  \rho-\sigma\right)
\right\}  =0$ for any Hermitian $B$, from which we can conclude that
$\rho=\sigma$.
\end{proof}

\begin{corollary}
[Joint quasi-convexity]\label{cor:joint-quasi-convexity}The sandwiched
relative R\'{e}nyi entropy $\widetilde{D}_{\alpha}\left(  A\|B\right)  $ is
jointly quasi-convex in its arguments for $\alpha\in(1,2]$, in the sense that%
\[
\widetilde{D}_{\alpha}\left(  A\|B\right)  \leq\max_{x}\widetilde{D}_{\alpha
}\left(  A_{x}\|B_{x}\right)  ,
\]
where $A=\sum_{x}p\left(  x\right)  A_{x}$ and $B=\sum_{x}p\left(  x\right)
B_{x}$.
\end{corollary}

\begin{proof}
This follows by employing joint convexity of $\widetilde{Q}_{\alpha}\left(
A\|B\right)  $:%
\begin{align*}
\widetilde{D}_{\alpha}\left(  A\|B\right)   &  =\frac{1}{\alpha-1}%
\log\widetilde{Q}_{\alpha}\left(  A\|B\right) \\
&  \leq\frac{1}{\alpha-1}\log\sum_{x}p\left(  x\right)  \widetilde{Q}_{\alpha
}\left(  A_{x}\|B_{x}\right) \\
&  \leq\frac{1}{\alpha-1}\log\max_{x}\widetilde{Q}_{\alpha}\left(
A_{x}\|B_{x}\right) \\
&  =\max_{x}\widetilde{D}_{\alpha}\left(  A_{x}\|B_{x}\right)  .
\end{align*}

\end{proof}

\begin{definition}
The von Neumann relative entropy for $A,B\in\mathcal{B}\left(  \mathcal{H}%
\right)  _{+}$ is defined as%
\[
D\left(  A\|B\right)  \equiv\left\{
\begin{array}
[c]{cc}%
\operatorname{Tr}\left\{  A\log A\right\}  -\operatorname{Tr}\left\{  A\log
B\right\}  & \text{if }\operatorname{supp}\left(  A\right)  \subseteq
\operatorname{supp}\left(  B\right) \\
\infty & \text{otherwise}%
\end{array}
\right.  .
\]

\end{definition}

\begin{proposition}
\label{prop:converge-to-vN} In the limit as $\alpha$ approaches one, the
sandwiched relative R\'{e}nyi entropy $\widetilde{D}_{\alpha}\left(
A\|B\right)  $ converges to the von\ Neumann relative entropy $D\left(
A\|B\right)  $ if $\operatorname{Tr}\{A\}=1$:%
\[
\lim_{\alpha\rightarrow1}\widetilde{D}_{\alpha}\left(  A\|B\right)  =D\left(
A\|B\right)  .
\]

\end{proposition}

\begin{proof}
A proof follows by exploiting some ideas of Carlen and Lieb \cite{CL08} and
Ogawa and Nagaoka \cite{ON99}. It suffices to show that%
\[
\left.  \frac{\partial}{\partial\alpha}\text{Tr}\left\{  \left(
B^{\frac{1-\alpha}{2\alpha}}AB^{\frac{1-\alpha}{2\alpha}}\right)  ^{\alpha
}\right\}  \right\vert _{\alpha=1}=\text{Tr}\left\{  A\log A\right\}
-\text{Tr}\left\{  A\log B\right\}  .
\]
This is because, in order to evaluate the limit, we require L'H\^{o}pital's
rule, so that%
\begin{align*}
\lim_{\alpha\rightarrow1}\widetilde{D}_{\alpha}\left(  A\|B\right)   &
=\lim_{\alpha\rightarrow1}\frac{1}{\alpha-1}\log\widetilde{Q}_{\alpha}\left(
A\|B\right) \\
&  =\lim_{\alpha\rightarrow1}\frac{1}{\widetilde{Q}_{\alpha}\left(
A\|B\right)  }\frac{\partial}{\partial\alpha}\widetilde{Q}_{\alpha}\left(
A\|B\right) \\
&  =\left.  \frac{\partial}{\partial\alpha}\text{Tr}\left\{  \left(
B^{\frac{1-\alpha}{2\alpha}}AB^{\frac{1-\alpha}{2\alpha}}\right)  ^{\alpha
}\right\}  \right\vert _{\alpha=1}.
\end{align*}
(In this proof, we will take $\log$ to denote the natural logarithm, but note
that the result follows simply by replacing the natural logarithm in both
definitions with the binary logarithm.) We assume that the support of $A$ is
contained in the support of $B$. Otherwise, there is nothing to prove since
both quantities are infinite.

Let us rewrite the expression inside the trace, using $\alpha=1+\varepsilon$,
as%
\[
\text{Tr}\left\{  \left(  B^{\frac{-\varepsilon}{2\left(  1+\varepsilon
\right)  }}AB^{\frac{-\varepsilon}{2\left(  1+\varepsilon\right)  }}\right)
^{1+\varepsilon}\right\}  .
\]
Furthermore, we can use two parameters $\varepsilon_{1}$ and $\varepsilon_{2}$
so that the above expression is a special case of%
\[
f\left(  \varepsilon_{1},\varepsilon_{2}\right)  \equiv\text{Tr}\left\{
\left(  B^{\frac{-\varepsilon_{1}}{2\left(  1+\varepsilon_{1}\right)  }%
}AB^{\frac{-\varepsilon_{1}}{2\left(  1+\varepsilon_{1}\right)  }}\right)
^{1+\varepsilon_{2}}\right\}  .
\]
We then have that%
\begin{align*}
\left.  \frac{\partial}{\partial\alpha}\text{Tr}\left\{  \left(
B^{\frac{1-\alpha}{2\alpha}}AB^{\frac{1-\alpha}{2\alpha}}\right)  ^{\alpha
}\right\}  \right\vert _{\alpha=1}  &  =\left.  \frac{\partial}{\partial
\varepsilon}f\left(  \varepsilon,\varepsilon\right)  \right\vert
_{\varepsilon=0}\\
&  =\left.  \frac{\partial}{\partial\varepsilon_{1}}f\left(  \varepsilon
_{1},0\right)  \right\vert _{\varepsilon_{1}=0}+\left.  \frac{\partial
}{\partial\varepsilon_{2}}f\left(  0,\varepsilon_{2}\right)  \right\vert
_{\varepsilon_{2}=0}.
\end{align*}
Consider the following Taylor expansions around $\varepsilon=0$%
\begin{align*}
X^{1+\varepsilon}  &  =X+\varepsilon X\log X+O\left(  \varepsilon^{2}\right)
,\\
X^{\frac{-\varepsilon}{1+\varepsilon}}  &  =I-\varepsilon\log X+O\left(
\varepsilon^{2}\right)  .
\end{align*}
From these, we calculate $f\left(  \varepsilon_{1},0\right)  $ as%
\begin{align*}
f\left(  \varepsilon_{1},0\right)   &  =\text{Tr}\left\{  B^{\frac
{-\varepsilon_{1}}{2\left(  1+\varepsilon_{1}\right)  }}AB^{\frac
{-\varepsilon_{1}}{2\left(  1+\varepsilon_{1}\right)  }}\right\} \\
&  =\text{Tr}\left\{  AB^{\frac{-\varepsilon_{1}}{1+\varepsilon_{1}}}\right\}
\\
&  =\text{Tr}\left\{  A\left(  I-\varepsilon_{1}\log B\right)  \right\}
+O\left(  \varepsilon_{1}^{2}\right) \\
&  =\text{Tr}\left\{  A\right\}  -\varepsilon_{1}\text{Tr}\left\{  A\log
B\right\}  +O\left(  \varepsilon_{1}^{2}\right)  .
\end{align*}
It then follows that%
\[
\left.  \frac{\partial}{\partial\varepsilon_{1}}f\left(  \varepsilon
_{1},0\right)  \right\vert _{\varepsilon_{1}=0}=-\text{Tr}\left\{  A\log
B\right\}  .
\]
Assuming that the support of $A$ is contained in the support of $B$, we then
calculate $f\left(  0,\varepsilon_{2}\right)  $ as%
\begin{align*}
f\left(  0,\varepsilon_{2}\right)   &  =\text{Tr}\left\{  A^{1+\varepsilon
_{2}}\right\} \\
&  =\text{Tr}\left\{  A\right\}  +\varepsilon_{2}\text{Tr}\left\{  A\log
A\right\}  +O\left(  \varepsilon_{2}^{2}\right)  .
\end{align*}
It then follows that%
\[
\left.  \frac{\partial}{\partial\varepsilon_{2}}f\left(  0,\varepsilon
_{2}\right)  \right\vert _{\varepsilon_{2}=0}=\text{Tr}\left\{  A\log
A\right\}  .
\]
Putting these together, we find that%
\[
\left.  \frac{\partial}{\partial\varepsilon}f\left(  \varepsilon
,\varepsilon\right)  \right\vert _{\varepsilon=0}=\text{Tr}\left\{  A\log
A\right\}  -\text{Tr}\left\{  A\log B\right\}  =D\left(  A\|B\right)  .
\]

\end{proof}

\subsection{Holevo-like quantity from the sandwiched R\'{e}nyi relative
entropy}

This section establishes a relation between $\widetilde{\chi}_{\alpha}\left(
\mathcal{N}\right)  $ and an $\alpha$-information radius quantity, defined
below. The development here gives an improvement to Lemma~I.3 in \cite{KW09},
such that we establish an equality rather than two inequalities, as seen by
comparing our Lemma~\ref{lem:chi-a-info-radius-equal} to Lemma~I.3 in
\cite{KW09}.

\begin{definition}
[$\alpha$-Holevo information]\label{def:alpha-holevo}By following the recipe
given in \eqref{eq:generalized-holevo}, we define the $\alpha$-Holevo
information of a channel $\mathcal{N}$ as follows:%
\[
\widetilde{\chi}_{\alpha}\left(  \mathcal{N}\right)  \equiv\max_{\left\{
p_{X}\left(  x\right)  ,\rho_{x}\right\}  }\widetilde{\chi}_{\alpha}\left(
\left\{  p_{X}\left(  x\right)  ,\mathcal{N}\left(  \rho_{x}\right)  \right\}
\right)  ,
\]
where%
\begin{align*}
\widetilde{\chi}_{\alpha}\left(  \left\{  p_{X}\left(  x\right)  ,\rho
_{x}\right\}  \right)   &  \equiv\min_{\sigma_{Q}}\widetilde{D}_{\alpha
}\left(  \rho_{XQ}\Vert\rho_{X}\otimes\sigma_{Q}\right)  ,\\
\rho_{XQ} &  \equiv\sum_{x}p_{X}\left(  x\right)  \left\vert x\right\rangle
\left\langle x\right\vert _{X}\otimes\left(  \rho_{x}\right)  _{Q}.
\end{align*}

\end{definition}

By exploiting the above definition and Definition~\ref{def:sandwich-rel-ent},
it follows that%
\begin{equation}
\widetilde{\chi}_{\alpha}\left(  \left\{  p_{X}\left(  x\right)  ,\rho
_{x}\right\}  \right)  =\min_{\sigma}\frac{1}{\alpha-1}\log\left[  \sum
_{x}p_{X}\left(  x\right)  \text{Tr}\left\{  \left(  \sigma^{\frac{1-\alpha
}{2\alpha}}\rho_{x}\sigma^{\frac{1-\alpha}{2\alpha}}\right)  ^{\alpha
}\right\}  \right]  . \label{eq:rewriting-a-holevo}%
\end{equation}

\begin{definition}
[$\alpha$-information radius]The $\alpha$-information radius of a channel
$\mathcal{N}$ \cite{MH11,C95,S69} is defined as%
\[
\widetilde{K}_{\alpha}\left(  \mathcal{N}\right)  \equiv\min_{\sigma}%
\max_{\rho}\widetilde{D}_{\alpha}\left(  \mathcal{N}\left(  \rho\right)
\|\sigma\right)  .
\]

\end{definition}

The reason that quantities like $\widetilde{K}_{\alpha}\left(  \mathcal{N}%
\right)  $ are often referred to as an \textquotedblleft information
radius\textquotedblright\ is that if we think of $\widetilde{D}_{\alpha}$ as a
distance measure\ (even though it is only a pseudo-distance), then it
quantifies the \textquotedblleft radius\textquotedblright\ of the possible
channel outputs $\mathcal{N}\left(  \rho\right)  $ with respect to the
distance measure $\widetilde{D}_{\alpha}$.

The following lemma is very helpful in analyzing whether $\widetilde{\chi
}_{\alpha}\left(  \mathcal{N}\right)  $ is additive for a given channel:

\begin{lemma}
\label{lem:chi-a-info-radius-equal}The $\alpha$-Holevo information
$\widetilde{\chi}_{\alpha}\left(  \mathcal{N}\right)  $ is equal to the
$\alpha$-information radius $\widetilde{K}_{\alpha}\left(  \mathcal{N}\right)
$ for $\alpha\in(1,2]$:%
\[
\widetilde{\chi}_{\alpha}\left(  \mathcal{N}\right)  =\widetilde{K}_{\alpha
}\left(  \mathcal{N}\right)  .
\]

\end{lemma}

\begin{proof}
We first prove the inequality $\widetilde{K}_{\alpha}\left(  \mathcal{N}%
\right)  \leq\widetilde{\chi}_{\alpha}\left(  \mathcal{N}\right)  $. Recalling
the definition of $\widetilde{Q}_{\alpha}$ from
Definition~\ref{def:sandwich-rel-ent}, consider that%
\begin{align*}
\widetilde{K}_{\alpha}\left(  \mathcal{N}\right)    & =\min_{\sigma}\max
_{\rho}\widetilde{D}_{\alpha}\left(  \mathcal{N}\left(  \rho\right)
\Vert\sigma\right)  \\
& =\min_{\sigma}\max_{\rho}\frac{1}{\alpha-1}\log\widetilde{Q}_{\alpha}\left(
\mathcal{N}\left(  \rho\right)  \Vert\sigma\right)  \\
& =\frac{1}{\alpha-1}\log\min_{\sigma}\max_{\rho}\widetilde{Q}_{\alpha}\left(
\mathcal{N}\left(  \rho\right)  \Vert\sigma\right)
\end{align*}
So now we focus on the $\widetilde{Q}_{\alpha}$ quantity and find that%
\begin{align}
\min_{\sigma}\max_{\rho}\widetilde{Q}_{\alpha}\left(  \mathcal{N}\left(
\rho\right)  \Vert\sigma\right)    & \leq\min_{\sigma}\sup_{\mu}\int
d\mu\left(  \rho\right)  \ \widetilde{Q}_{\alpha}\left(  \mathcal{N}\left(
\rho\right)  \Vert\sigma\right)  \nonumber \\
& =\sup_{\mu}\min_{\sigma}\int d\mu\left(  \rho\right)  \ \widetilde
{Q}_{\alpha}\left(  \mathcal{N}\left(  \rho\right)  \Vert\sigma\right)  \nonumber\\
& =\max_{\left\{  p_{X}\left(  x\right)  ,\rho_{x}\right\}  }\min_{\sigma}%
\sum_{x}p_{X}\left(  x\right)  \widetilde{Q}_{\alpha}\left(  \mathcal{N}%
\left(  \rho_{x}\right)  \Vert\sigma\right)  \nonumber\\
& =\max_{\left\{  p_{X}\left(  x\right)  ,\rho_{x}\right\}  }\min_{\sigma_{B}%
}\widetilde{Q}_{\alpha}\left(  \rho_{XB}\Vert\rho_{X}\otimes\sigma_{B}\right)
\label{eq:K<=chi}
\end{align}
The first inequality follows by taking a supremum over all probability
measures $\mu$ on the set of all states $\rho$. The first equality is a result of
applying the Sion minimax theorem \cite{S58}---we can do so because the function $\int
d\mu\left(  \rho\right)  \ \widetilde{Q}_{\alpha}\left(  \mathcal{N}\left(
\rho\right)  \Vert\sigma\right)  $ is linear in the probability measure $\mu$
and convex in states $\sigma$. Convexity of $\widetilde{Q}_{\alpha}\left(
\mathcal{N}\left(  \rho\right)  \Vert\sigma\right)  $ in $\sigma$ follows
because
\[
\widetilde{Q}_{\alpha}\left(  \mathcal{N}\left(  \rho\right)  \Vert
\sigma\right)  =\text{Tr}\left\{  \left(  \left[  \mathcal{N}\left(
\rho\right)  \right]  ^{1/2}\sigma^{\left(  1-\alpha\right)  /\alpha}\left[
\mathcal{N}\left(  \rho\right)  \right]  ^{1/2}\right)  ^{\alpha}\right\}  ,
\]
$x^{\left(  1-\alpha\right)  /\alpha}$ is operator convex for $\alpha\in(1,2]$
and $x^{\alpha}$ is operator convex for $\alpha\in(1,2]$. The second equality
follows by an application of the Fenchel-Eggleston-Caratheodory theorem
(see \cite{EK12}, for example):\ the
function $\widetilde{Q}_{\alpha}\left(  \mathcal{N}\left(  \rho\right)
\Vert\sigma\right)  $ is continuous in $\rho$, which is a density operator
acting on a $d$-dimensional Hilbert space, so that to each $\mu$, there exists
a probability distribution $p_{X}\left(  x\right)  $ on no more than $d^{2}$
letters such that%
\[
\int d\mu\left(  \rho\right)  \ \widetilde{Q}_{\alpha}\left(  \mathcal{N}%
\left(  \rho\right)  \Vert\sigma\right)  =\sum_{x}p_{X}\left(  x\right)
\widetilde{Q}_{\alpha}\left(  \mathcal{N}\left(  \rho_{x}\right)  \Vert
\sigma\right)  .
\]
The last equality in \eqref{eq:K<=chi} follows from the properties of $\widetilde{Q}_{\alpha
}$ and by defining%
\[
\rho_{XB}\equiv\sum_{x}p_{X}\left(  x\right)  \left\vert x\right\rangle
\left\langle x\right\vert _{X}\otimes\left[  \mathcal{N}\left(  \rho
_{x}\right)  \right]  _{B}.
\]
So we can then conclude that%
\begin{align*}
\widetilde{K}_{\alpha}\left(  \mathcal{N}\right)    & \leq\frac{1}{\alpha
-1}\log\max_{\left\{  p_{X}\left(  x\right)  ,\rho_{x}\right\}  }\min
_{\sigma_{B}}\widetilde{Q}_{\alpha}\left(  \rho_{XB}\Vert\rho_{X}\otimes
\sigma_{B}\right)  \\
& =\max_{\left\{  p_{X}\left(  x\right)  ,\rho_{x}\right\}  }\min_{\sigma_{B}%
}\frac{1}{\alpha-1}\log\widetilde{Q}_{\alpha}\left(  \rho_{XB}\Vert\rho
_{X}\otimes\sigma_{B}\right)  \\
& =\widetilde{\chi}_{\alpha}\left(  \mathcal{N}\right)  .
\end{align*}

The proof of the other inequality $\widetilde{K}_{\alpha}\left(
\mathcal{N}\right)  \geq\widetilde{\chi}_{\alpha}\left(  \mathcal{N}\right)
$\ is simpler. Consider that
\begin{align*}
\widetilde{\chi}_{\alpha}\left(  \mathcal{N}\right)   &  =\max_{\left\{
p_{X}\left(  x\right)  ,\rho_{x}\right\}  }\min_{\sigma}\widetilde{D}_{\alpha
}\left(  \rho_{XB}\Vert\rho_{X}\otimes\sigma\right)  \\
&  \leq\max_{\left\{  p_{X}\left(  x\right)  ,\rho_{x}\right\}  }\widetilde
{D}_{\alpha}\left(  \rho_{XB}\Vert\rho_{X}\otimes\sigma\right)  \\
&  \leq\max_{x}\widetilde{D}_{\alpha}\left(  \left\vert x\right\rangle
\left\langle x\right\vert \otimes\mathcal{N}\left(  \rho_{x}\right)
\Vert\left\vert x\right\rangle \left\langle x\right\vert \otimes\sigma\right)
\\
&  =\max_{x}\widetilde{D}_{\alpha}\left(  \mathcal{N}\left(  \rho_{x}\right)
\Vert\sigma\right)  \\
&  \leq\max_{\rho}\widetilde{D}_{\alpha}\left(  \mathcal{N}\left(
\rho\right)  \Vert\sigma\right)  .
\end{align*}
The second inequality follows from joint quasi-convexity of $\widetilde
{D}_{\alpha}$ (Lemma~\ref{cor:joint-quasi-convexity}). Since the above
inequality holds for all states $\sigma$, we can conclude that $\widetilde
{K}_{\alpha}\left(  \mathcal{N}\right)  \geq\widetilde{\chi}_{\alpha}\left(
\mathcal{N}\right)  $. (This last realization is what allows for the
improvement over Lemma I.3 in \cite{KW09}.)
\end{proof}

\begin{remark}
The above proof unchanged demonstrates that%
\[
\chi_{\alpha}\left(  \mathcal{N}\right)  =K_{\alpha}\left(  \mathcal{N}%
\right)  ,
\]
where these quantities are defined in the same way as $\widetilde{\chi
}_{\alpha}\left(  \mathcal{N}\right)  $ and $\widetilde{K}_{\alpha}\left(
\mathcal{N}\right)  $, except through the traditional R\'enyi relative entropy
defined in \eqref{eq:traditional-renyi-entropy}.
\end{remark}

\subsection{The sandwiched R\'{e}nyi relative entropy is induced by a norm}

We define the sandwiched $\alpha$-norm of an operator $A$\ by%
\[
\left\Vert A\right\Vert _{\alpha,X}\equiv\left\Vert X^{1/2 }A X^{1/2}%
\right\Vert _{\alpha},
\]
for any positive operator $X$ and where%
\[
\left\Vert B\right\Vert _{\alpha}\equiv\ \text{Tr}\{ ( \sqrt{B^{\dag}B})
^{\alpha}\} ^{1/\alpha}.
\]
With these definitions and that in (\ref{eq:alt-renyi-rel-ent}), it is easy to
see that for $\alpha>1$%
\begin{align}
\widetilde{K}_{\alpha}\left(  \mathcal{N}\right)   &  \equiv\min_{\sigma}%
\max_{\rho}\frac{\alpha}{\alpha-1}\log\left\Vert \mathcal{N}\left(
\rho\right)  \right\Vert _{\alpha,\sigma^{\frac{1-\alpha}{\alpha}}}\nonumber\\
&  =\min_{\sigma}\frac{\alpha}{\alpha-1}\log\max_{\rho}\left\Vert
\mathcal{N}\left(  \rho\right)  \right\Vert _{\alpha,\sigma^{\frac{1-\alpha
}{\alpha}}}. \label{eq:rewrite-K-a}%
\end{align}
This reformulation in terms of the sandwiched $\alpha$-norm will make it
easier to see that $\widetilde{\chi}_{\alpha}$ is subadditive for the class of
entanglement-breaking channels.

\section{Bounding the success probability with the sandwiched relative
R\'{e}nyi entropy}

\label{sec:success-prob-bnd-sandwich}Combining the results of the previous two
sections (i.e., the bound in (\ref{eq:p_succ-bound}) and the fact that the
sandwiched R\'{e}nyi relative entropy is a generalized divergence), we find
the following bound on the success probability for any rate $R$ scheme for
classical communication over $n$ uses of a quantum channel$~\mathcal{N}$:%
\begin{equation}
p_{\text{succ}}\leq2^{-n\left(  \frac{\alpha-1}{\alpha}\right)  \left(
R-\frac{1}{n}\widetilde{\chi}_{\alpha}\left(  \mathcal{N}^{\otimes n}\right)
\right)  }, \label{eq:error-bound-a-holevo}%
\end{equation}
for all $\alpha\in(1,2]$. Indeed, since the divergence $\widetilde{D}_{\alpha
}$ satisfies all of the requirements from Section~\ref{sec:sharma}, we find
the following bound%
\[
\widetilde{\chi}_{\alpha}\left(  \mathcal{N}^{\otimes n}\right)
\geq\widetilde{\delta}_{\alpha}\left(  \varepsilon\|1-2^{-nR}\right)  ,
\]
where $\widetilde{\delta}_{\alpha}$ is the classical divergence induced from
$\widetilde{D}_{\alpha}$. Since the following inequality holds for $\alpha>1$%
\begin{align*}
\widetilde{\delta}_{\alpha}\left(  \varepsilon\|1-2^{-nR}\right)   &
=\frac{1}{\alpha-1}\log\left(  \varepsilon^{\alpha}\left(  1-2^{-nR}\right)
^{1-\alpha}+\left(  1-\varepsilon\right)  ^{\alpha}\left(  2^{-nR}\right)
^{1-\alpha}\right) \\
&  \geq\frac{1}{\alpha-1}\log\left(  \left(  1-\varepsilon\right)  ^{\alpha
}\left(  2^{-nR}\right)  ^{1-\alpha}\right) \\
&  =\frac{\alpha}{\alpha-1}\log\left(  1-\varepsilon\right)  +nR,
\end{align*}
we arrive at (\ref{eq:error-bound-a-holevo}). Thus, we have now reduced the
proof of the strong converse to the subadditivity of the quantity
$\widetilde{\chi}_{\alpha}\left(  \mathcal{N}^{\otimes n}\right)  $.

\section{Subadditivity of the $\alpha$-information radius for
entanglement-breaking channels}

The main result of this section is that $\widetilde{\chi}_{\alpha}\left(
\mathcal{N}^{\otimes n}\right)  \leq n\widetilde{\chi}_{\alpha}\left(
\mathcal{N}\right)  $ whenever $\mathcal{N}$ is an entanglement-breaking
channel. We start by recalling a definition and a theorem:

\begin{definition}
The maximum output $\alpha$-norm of a completely positive\ map $\mathcal{M}$
is defined as%
\[
\nu_{\alpha}\left(  \mathcal{M}\right)  \equiv\max_{\rho}\left\Vert
\mathcal{M}\left(  \rho\right)  \right\Vert _{\alpha}.
\]

\end{definition}

\begin{theorem}
[\cite{K03,H06}]\label{thm:mult-ent-break}The maximum output $\alpha$-norm is
multiplicative for a completely-positive entanglement-breaking map
$\mathcal{M}_{\operatorname{EB}}$ and an arbitrary completely positive map
$\mathcal{M}$ for all $\alpha\geq1$:%
\[
\nu_{\alpha}(\mathcal{M}_{\operatorname{EB}}\otimes\mathcal{M})=\nu_{\alpha
}(\mathcal{M}_{\operatorname{EB}})\ \nu_{\alpha}\left(  \mathcal{M}\right)  .
\]

\end{theorem}

The following subadditivity relation then results from the above theorem:

\begin{theorem}
\label{thm:subadditivity-EB}For an entanglement-breaking channel
$\mathcal{N}_{\operatorname{EB}}$ and any other channel $\mathcal{N}$ and for
all $\alpha\in(1,2]$, the following subadditivity relation holds%
\[
\widetilde{\chi}_{\alpha}(\mathcal{N}_{\operatorname{EB}}\otimes
\mathcal{N})\leq\widetilde{\chi}_{\alpha}(\mathcal{N}_{\operatorname{EB}%
})+\widetilde{\chi}_{\alpha}(\mathcal{N}).
\]

\end{theorem}

\begin{proof}
We proceed with just a few steps:%
\begin{align*}
\widetilde{\chi}_{\alpha}(\mathcal{N}_{\operatorname{EB}}\otimes\mathcal{N})
&  =\widetilde{K}_{\alpha}(\mathcal{N}_{\operatorname{EB}}\otimes
\mathcal{N})\\
&  =\min_{\sigma_{B_{1}B_{2}}}\frac{\alpha}{\alpha-1}\log\max_{\rho
_{A_{1}A_{2}}}\left\Vert (\mathcal{N}_{\operatorname{EB}}\otimes
\mathcal{N})\left(  \rho_{A_{1}A_{2}}\right)  \right\Vert _{\alpha
,\sigma_{B_{1}B_{2}}^{\left(  1-\alpha\right)  /\alpha}}\\
&  \leq\min_{\sigma_{B_{1}}\otimes\sigma_{B_{2}}}\frac{\alpha}{\alpha-1}%
\log\max_{\rho_{A_{1}A_{2}}}\left\Vert (\mathcal{N}_{\operatorname{EB}}%
\otimes\mathcal{N})\left(  \rho_{A_{1}A_{2}}\right)  \right\Vert
_{\alpha,\sigma_{B_{1}}^{\left(  1-\alpha\right)  /\alpha}\otimes\sigma
_{B_{2}}^{\left(  1-\alpha\right)  /\alpha}}\\
&  \leq\min_{\sigma_{B_{1}}\otimes\sigma_{B_{2}}}\frac{\alpha}{\alpha-1}%
\log\left[  \max_{\rho_{A_{1}}}\left\Vert \mathcal{N}_{\operatorname{EB}%
}\left(  \rho_{A_{1}}\right)  \right\Vert _{\alpha,\sigma_{B_{1}}^{\left(
1-\alpha\right)  /\alpha}}\max_{\rho_{A_{2}}}\left\Vert \mathcal{N}\left(
\rho_{A_{2}}\right)  \right\Vert _{\alpha,\sigma_{B_{2}}^{\left(
1-\alpha\right)  /\alpha}}\right] \\
&  =\min_{\sigma_{B_{1}}\otimes\sigma_{B_{2}}}\frac{\alpha}{\alpha-1}\left[
\log\max_{\rho_{A_{1}}}\left\Vert \mathcal{N}_{\operatorname{EB}}\left(
\rho_{A_{1}}\right)  \right\Vert _{\alpha,\sigma_{B_{1}}^{\left(
1-\alpha\right)  /\alpha}}+\log\max_{\rho_{A_{2}}}\left\Vert \mathcal{N}%
\left(  \rho_{A_{2}}\right)  \right\Vert _{\alpha,\sigma_{B_{2}}^{\left(
1-\alpha\right)  /\alpha}}\right] \\
&  =\widetilde{K}_{\alpha}(\mathcal{N}_{\operatorname{EB}})+\widetilde
{K}_{\alpha}(\mathcal{N})\\
&  =\widetilde{\chi}_{\alpha}(\mathcal{N}_{\operatorname{EB}})+\widetilde
{\chi}_{\alpha}(\mathcal{N}).
\end{align*}
The first equality follows from Lemma~\ref{lem:chi-a-info-radius-equal}. The
second equality follows from the observation in (\ref{eq:rewrite-K-a}).\ The
first inequality follows by minimizing over tensor-product states rather than
general states. The second inequality follows from the observation in
Remark~\ref{rem:ent-closure} (that an entanglement-breaking map conjugated by
a positive operator $\sigma_{B_{1}}^{(1-\alpha)/2\alpha}$ is still an
entanglement-breaking map) and from Theorem~\ref{thm:mult-ent-break}. The last
few equalities follow by applying the logarithm and from definitions.
\end{proof}

The above subadditivity relation and an inductive argument are sufficient for
us to conclude the following corollary:

\begin{corollary}
\label{cor:subadd-EB}For an entanglement-breaking channel $\mathcal{N}%
_{\operatorname{EB}}$, for all $\alpha\in(1,2]$, and for any positive integer
$n$, we have the following subadditivity relation:%
\[
\widetilde{\chi}_{\alpha}(\mathcal{N}_{\operatorname{EB}}^{\otimes n})\leq n\,
\widetilde{\chi}_{\alpha}(\mathcal{N}_{\operatorname{EB}}).
\]

\end{corollary}

\section{Final steps for the strong converse for entanglement-breaking
channels}

\label{sec:final-steps-EB}Returning to (\ref{eq:error-bound-a-holevo}), the
subadditivity relation from Corollary~\ref{cor:subadd-EB}\ allows us to
conclude the following upper bound on the success probability when
communicating over an entanglement-breaking channel $\mathcal{N}%
_{\operatorname{EB}}$:%
\begin{equation}
p_{\text{succ}}\leq2^{-n\left(  \frac{\alpha-1}{\alpha}\right)  \left(
R-\widetilde{\chi}_{\alpha}\left(  \mathcal{N}_{\operatorname{EB}}\right)
\right)  }. \label{eq:improved-bound-succ-prob}%
\end{equation}
It follows by a standard argument \cite{ON99,KW09}\ that if $R>\chi
(\mathcal{N}_{\operatorname{EB}})$, then the success probability decreases
exponentially fast in $n$ to zero. That is, we can analyze the derivative of
$K_{\alpha}\left(  \mathcal{N}_{\operatorname{EB}}\right)  $ with respect to
$\alpha$ and as $\alpha\rightarrow1$, $K_{\alpha}\left(  \mathcal{N}%
_{\operatorname{EB}}\right)  $ approaches $\min_{\sigma}\max_{\rho}D\left(
\mathcal{N}_{\operatorname{EB}}\left(  \rho\right)  \|\sigma\right)  $ which
we know is equal to $\chi\left(  \mathcal{N}_{\operatorname{EB}}\right)  $
\cite{OPW97,SW01}. If $R > \chi(\mathcal{N}_{\operatorname{EB}})$, one can
always find an $\alpha$ close enough to one such that the exponent
\[
\left(  \frac{\alpha-1}{\alpha}\right)  \left(  R-\widetilde{\chi}_{\alpha
}\left(  \mathcal{N}_{\operatorname{EB}}\right)  \right)  >0.
\]
One could then take a supremum over all $\alpha\in(1,2]$ to optimize the
exponent. We point the reader to Section~6 of \cite{GW13}\ for additional
details of this standard argument. From this line of reasoning, we can
conclude the strong converse for entanglement-breaking channels.

However, we can also prove this result with a different approach. The
resulting bound still gives an exponential decay of the success probability,
but the approach above gives a stronger decay since it includes an
optimization over the R\'enyi parameter $\alpha$. Consider the following
inequality from Lemma~6.3 of Ref.~\cite{T12}:%
\begin{equation}
D_{\alpha}\left(  \rho\|\sigma\right)  \leq D\left(  \rho\|\sigma\right)
+4\left(  \alpha-1\right)  \left(  \log\nu\right)  ^{2},
\label{eq:fq-aep-relation}%
\end{equation}
where%
\begin{align}
1  &  <\alpha<1+\frac{\log3}{4\log\nu},\label{eq:cond-on-alpha}\\
\nu &  =2^{\frac{1}{2}D_{3/2}\left(  \rho\|\sigma\right)  }+2^{-\frac{1}%
{2}D_{1/2}\left(  \rho\|\sigma\right)  }+1.\nonumber
\end{align}
Combining the inequality above and in (\ref{eq:ineq-d-tilde-alpha-d-alpha}),
we find that%
\begin{equation}
\widetilde{D}_{\alpha}\left(  \rho\|\sigma\right)  \leq D\left(  \rho
\|\sigma\right)  +4\left(  \alpha-1\right)  \left(  \log\nu\right)  ^{2}.
\label{eq:bnd-with-vN-rel-ent}%
\end{equation}
We can use this bound to deduce the strong converse.

Consider the information radius \cite{OPW97,SW01}:%
\[
\min_{\sigma}\max_{\rho}D\left(  \mathcal{N}_{\operatorname{EB}}\left(
\rho\right)  \|\sigma\right)  =\chi\left(  \mathcal{N}_{\operatorname{EB}%
}\right)  .
\]
We know that there is an optimal value of $\sigma$ for the above quantity, and
let us call it $\sigma^{\ast}$. Furthermore, we know that%
\[
\max_{\rho}D\left(  \mathcal{N}_{\operatorname{EB}}\left(  \rho\right)
\|\sigma^{\ast}\right)
\]
is a finite number (because it is equal to $\chi(\mathcal{N}%
_{\operatorname{EB}})$). Thus, the support of $\mathcal{N}_{\operatorname{EB}%
}\left(  \rho\right)  $ is contained in the support of $\sigma^{\ast}$ for all
$\rho$---otherwise, there would be some $\rho$ that could make the above
quantity infinite. So using (\ref{eq:fq-aep-relation}), we have the following
inequality holding for all $\rho$:%
\begin{equation}
\widetilde{D}_{\alpha}\left(  \mathcal{N}_{\operatorname{EB}}\left(
\rho\right)  \|\sigma^{\ast}\right)  \leq D\left(  \mathcal{N}%
_{\operatorname{EB}}\left(  \rho\right)  \|\sigma^{\ast}\right)  +4\left(
\alpha-1\right)  \left(  \log\nu\right)  ^{2}, \label{eq:almost-there}%
\end{equation}
where%
\[
\nu=2^{\frac{1}{2}D_{3/2}\left(  \mathcal{N}_{\operatorname{EB}}\left(
\rho\right)  \|\sigma^{\ast}\right)  }+2^{-\frac{1}{2}D_{1/2}\left(
\mathcal{N}_{\operatorname{EB}}\left(  \rho\right)  \|\sigma^{\ast}\right)
}+1.
\]
Since%
\[
2^{-\frac{1}{2}D_{1/2}\left(  \mathcal{N}_{\operatorname{EB}}\left(
\rho\right)  \|\sigma^{\ast}\right)  }=\text{Tr}\left\{  \sqrt{\mathcal{N}%
_{\operatorname{EB}}\left(  \rho\right)  }\sqrt{\sigma^{\ast}}\right\}
\leq1,
\]
it follows that%
\[
\nu\leq2^{\frac{1}{2}D_{3/2}\left(  \mathcal{N}_{\operatorname{EB}}\left(
\rho\right)  \|\sigma^{\ast}\right)  }+2.
\]
Also, since the support of $\mathcal{N}_{\operatorname{EB}}\left(
\rho\right)  $ is contained in the support of $\sigma^{\ast}$ for all $\rho$,
it follows that $D_{3/2}\left(  \mathcal{N}_{\operatorname{EB}}\left(
\rho\right)  \|\sigma^{\ast}\right)  <\infty$, so that%
\[
\nu\leq c\left(  \mathcal{N}_{\operatorname{EB}}\right)  <\infty,
\]
where $c\left(  \mathcal{N}_{\operatorname{EB}}\right)  $ is some constant
that depends on the channel $\mathcal{N}_{\operatorname{EB}}$ (we can pick it
to be independent of $\rho$ as well). Combining with (\ref{eq:almost-there}),
we find that%
\[
\max_{\rho}\widetilde{D}_{\alpha}\left(  \mathcal{N}_{\operatorname{EB}%
}\left(  \rho\right)  \|\sigma^{\ast}\right)  \leq\max_{\rho}D\left(
\mathcal{N}_{\operatorname{EB}}\left(  \rho\right)  \|\sigma^{\ast}\right)
+4\left(  \alpha-1\right)  \left(  \log c\left(  \mathcal{N}%
_{\operatorname{EB}}\right)  \right)  ^{2}.
\]
Taking one more minimization and recalling the choice of $\sigma^{\ast}$
finally gives that%
\[
\min_{\sigma}\max_{\rho}\widetilde{D}_{\alpha}\left(  \mathcal{N}%
_{\operatorname{EB}}\left(  \rho\right)  \|\sigma\right)  \leq\min_{\sigma
}\max_{\rho}D\left(  \mathcal{N}_{\operatorname{EB}}\left(  \rho\right)
\|\sigma\right)  +4\left(  \alpha-1\right)  \left(  \log c\left(
\mathcal{N}_{\operatorname{EB}}\right)  \right)  ^{2},
\]
which is equivalent to%
\begin{equation}
\widetilde{K}_{\alpha}\left(  \mathcal{N}_{\operatorname{EB}}\right)  \leq
\chi\left(  \mathcal{N}_{\operatorname{EB}}\right)  +4\left(  \alpha-1\right)
\left(  \log c\left(  \mathcal{N}_{\operatorname{EB}}\right)  \right)  ^{2}.
\label{eq:bound-on-K-alpha}%
\end{equation}

Finally, assume that $R>\chi\left(  \mathcal{N}_{\operatorname{EB}}\right)  $.
We choose $\alpha$ as follows:%
\[
\alpha=1+\min\left\{  \frac{\log3}{4\log c\left(  \mathcal{N}%
_{\operatorname{EB}}\right)  },\frac{R-\chi\left(  \mathcal{N}%
_{\operatorname{EB}}\right)  }{8\left(  \log c\left(  \mathcal{N}%
_{\operatorname{EB}}\right)  \right)  ^{2}},1\right\}  ,
\]
so that the following inequality holds%
\[
\chi\left(  \mathcal{N}_{\operatorname{EB}}\right)  +\left(  \alpha-1\right)
\left(  \log c\left(  \mathcal{N}_{\operatorname{EB}}\right)  \right)
^{2}\leq\frac{1}{2}\left(  R+\chi\left(  \mathcal{N}_{\operatorname{EB}%
}\right)  \right)  .
\]
(Furthermore, it is reasonable for us to assume that $R$ is close enough to
$\chi\left(  \mathcal{N}_{\operatorname{EB}}\right)  $ so that $\alpha$ is
actually equal to $1+\left[  R-\chi\left(  \mathcal{N}_{\operatorname{EB}%
}\right)  \right]  /8\left(  \log c\left(  \mathcal{N}_{\operatorname{EB}%
}\right)  \right)  ^{2}$.) Using the bounds in
(\ref{eq:improved-bound-succ-prob}) and (\ref{eq:bound-on-K-alpha}), we then
obtain the following bound on the success probability for any classical
communication protocol over an entanglement-breaking channel:%
\begin{align}
p_{\text{succ}}  &  \leq2^{-n\left(  \frac{\alpha-1}{\alpha}\right)  \left(
R-\widetilde{\chi}_{\alpha}\left(  \mathcal{N}_{\operatorname{EB}}\right)
\right)  }\nonumber\\
&  =2^{-n\left(  \frac{\alpha-1}{\alpha}\right)  \left[  R-\widetilde
{K}_{\alpha}\left(  \mathcal{N}_{\operatorname{EB}}\right)  \right]
}\nonumber\\
&  \leq2^{-n\left(  \frac{\alpha-1}{\alpha}\right)  \left[  R-\left[
\chi\left(  \mathcal{N}_{\operatorname{EB}}\right)  +4\left(  \alpha-1\right)
\left(  \log c\left(  \mathcal{N}_{\operatorname{EB}}\right)  \right)
^{2}\right]  \right]  }\nonumber\\
&  \leq2^{-n\left(  \frac{\alpha-1}{2}\right)  \left[  R-\left[  \frac{1}%
{2}\left(  R+\chi\left(  \mathcal{N}_{\operatorname{EB}}\right)  \right)
\right]  \right]  }\nonumber\\
&  =2^{-n\left(  \frac{\alpha-1}{4}\right)  \left[  R-\chi\left(
\mathcal{N}_{\operatorname{EB}}\right)  \right]  }\nonumber\\
&  \leq2^{-n\left[  R-\chi\left(  \mathcal{N}_{\operatorname{EB}}\right)
\right]  ^{2}/32\left(  \log c\left(  \mathcal{N}_{\operatorname{EB}}\right)
\right)  ^{2}}.
\end{align}
Thus, in the case that $R>\chi\left(  \mathcal{N}_{\operatorname{EB}}\right)
$, the success probability converges exponentially fast to zero.

One might be concerned about our restriction to rates near $\chi\left(
\mathcal{N}_{\operatorname{EB}}\right)  $, but it is also easy to see that
choosing $\alpha=1+\frac{1}{\sqrt{n}}$ recovers the bound%
\[
p_{\text{succ}}\leq2^{-\sqrt{n}\left(  \frac{1}{1+1/\sqrt{n}}\right)  \left[
R-\left[  \chi\left(  \mathcal{N}_{\operatorname{EB}}\right)  +\frac{4}%
{\sqrt{n}}\left(  \log c\left(  \mathcal{N}_{\operatorname{EB}}\right)
\right)  ^{2}\right]  \right]  },
\]
which decays to zero exponentially fast in $\sqrt{n}$ for any rate
$R>\chi\left(  \mathcal{N}_{\operatorname{EB}}\right)  $.

\subsection{Prior results on particular covariant channels follow as a special
case}

We remark briefly on how the prior results in Ref.~\cite{KW09} follow as a
special case of our approach. There, Koenig and Wehner showed that the strong
converse theorem holds for all covariant channels with an additive minimum
output R\'{e}nyi entropy. For these channels, they proved that%
\[
\chi_{\alpha}\left(  \mathcal{N}^{\otimes n}\right)  =n\left[  \log
d-H_{\alpha}^{\min}\left(  \mathcal{N}\right)  \right]  ,
\]
where the minimum output R\'{e}nyi entropy of a channel is defined as%
\begin{align*}
H_{\alpha}^{\min}\left(  \mathcal{N}\right)   &  \equiv\min_{\rho}H_{\alpha
}\left(  \mathcal{N}\left(  \rho\right)  \right)  ,\\
H_{\alpha}\left(  \sigma\right)   &  \equiv\frac{1}{1-\alpha}\log
\text{Tr}\left\{  \sigma^{\alpha}\right\}  .
\end{align*}
By following a development similar to that in the previous section, the strong
converse for these channels follows.

To recover their result, we can modify the proof of
Theorem~\ref{thm:subadditivity-EB} as follows:%
\begin{align*}
\widetilde{\chi}_{\alpha}(\mathcal{N}_{1}\otimes\mathcal{N}_{2})  &
=\widetilde{K}_{\alpha}(\mathcal{N}_{1}\otimes\mathcal{N}_{2})\\
&  =\min_{\sigma_{B_{1}B_{2}}}\frac{\alpha}{\alpha-1}\log\max_{\rho
_{A_{1}A_{2}}}\left\Vert (\mathcal{N}_{1}\otimes\mathcal{N}_{2})\left(
\rho_{A_{1}A_{2}}\right)  \right\Vert _{\alpha,\sigma_{B_{1}B_{2}}^{\left(
1-\alpha\right)  /\alpha}}\\
&  \leq\frac{\alpha}{\alpha-1}\log\max_{\rho_{A_{1}A_{2}}}\left\Vert
(\mathcal{N}_{1}\otimes\mathcal{N}_{2})\left(  \rho_{A_{1}A_{2}}\right)
\right\Vert _{\alpha,\pi_{B_{1}}^{\left(  1-\alpha\right)  /\alpha}\otimes
\pi_{B_{2}}^{\left(  1-\alpha\right)  /\alpha}}\\
&  =\log d_{1}+\log d_{2}-H_{\alpha}^{\min}\left(  \mathcal{N}_{1}%
\otimes\mathcal{N}_{2}\right)  ,
\end{align*}
where we denote the maximally mixed state by $\pi$. The inequality follows
simply by making the suboptimal choice of setting $\sigma_{B_{1}B_{2}}$ to be
the maximally mixed state. Thus, if $H_{\alpha}^{\min}\left(  \mathcal{N}%
_{1}\otimes\mathcal{N}_{2}\right)  =H_{\alpha}^{\min}\left(  \mathcal{N}%
_{1}\right)  +H_{\alpha}^{\min}\left(  \mathcal{N}_{2}\right)  $ for some
particular channels $\mathcal{N}_{1}$ and $\mathcal{N}_{2}$, we can then
conclude additivity of $\widetilde{\chi}_{\alpha}(\mathcal{N}_{1}%
\otimes\mathcal{N}_{2})$. All the classes of channels considered by Koenig and
Wehner have the property that the minimum output entropy of the channel and
any other channel is additive. Thus, one can conclude additivity of
$H_{\alpha}^{\min}\left(  \mathcal{N}^{\otimes n}\right)  $ by an inductive
argument that is the same as what we used in Corollary~\ref{cor:subadd-EB}.
The rest of the proof follows easily after establishing subadditivity of
$\widetilde{\chi}_{\alpha}\left(  \mathcal{N}^{\otimes n}\right)  $.

The above development in fact shows that we obtain a strong converse rate of
$\log d-H^{\min}(\mathcal{N})$ for any channel for which its minimum output
R\'{e}nyi entropy is additive for all $\alpha\geq1$. (In the above, $H^{\min
}(\mathcal{N})$ denotes the minimum output von Neumann entropy of the channel.)

\section{Strong converse for the classical capacity of Hadamard channels}

We now prove that the strong converse holds for the classical capacity of
Hadamard channels. This result follows from the following theorem, along with
some additional arguments:

\begin{theorem}
[\cite{KMNR07,H06}]\label{thm:mult-from-complements}If the maximum output
$\alpha$-norm is multiplicative for one pair of completely positive maps
$\mathcal{M}_{1}$\ and $\mathcal{M}_{2}$:%
\[
\nu_{\alpha}(\mathcal{M}_{1}\otimes\mathcal{M}_{2})=\nu_{\alpha}%
(\mathcal{M}_{1})\ \nu_{\alpha}\left(  \mathcal{M}_{2}\right)  ,
\]
then the same is true for their respective complementary maps $\mathcal{M}%
_{1}^{c}$\ and $\mathcal{M}_{2}^{c}$:%
\[
\nu_{\alpha}(\mathcal{M}_{1}^{c}\otimes\mathcal{M}_{2}^{c})=\nu_{\alpha
}(\mathcal{M}_{1}^{c})\ \nu_{\alpha}\left(  \mathcal{M}_{2}^{c}\right)  .
\]

\end{theorem}

\begin{definition}
Given a given channel $\mathcal{N}$ and a state $\sigma$ on the output space
of $\mathcal{N}$, let $\widetilde{K}_{\alpha}^{[\sigma]}(\mathcal{N})$ denote
the $\alpha$-information radius of the channel around $\sigma$:%
\begin{equation}
\widetilde{K}_{\alpha}^{[\sigma]}(\mathcal{N})\equiv\max_{\rho}\widetilde
{D}_{\alpha}(\mathcal{N}(\rho)\Vert\sigma).
\end{equation}

\end{definition}

Note that by definition, $\widetilde{K}_{\alpha}(\mathcal{N})=\min_{\sigma
}\widetilde{K}_{\alpha}^{[\sigma]}(\mathcal{N})$.


By a similar development as in Section~\ref{sec:success-prob-bnd-sandwich}, we
find that the following inequality holds for any code of rate $R$ with success
probability $1-\varepsilon$ that uses the channel $n$ times:%
\begin{align*}
\frac{1}{\alpha-1}\log\left(  \left(  1-\varepsilon\right)  ^{\alpha}\left(
2^{-nR}\right)  ^{1-\alpha}\right)   &  \leq\widetilde{\chi}_{\alpha}\left(
\mathcal{N}^{\otimes n}\right) \\
&  =\widetilde{K}_{\alpha}(\mathcal{N}^{\otimes n})\\
&  \leq\widetilde{K}_{\alpha}^{\left[  \sigma^{\otimes n}\right]
}(\mathcal{N}^{\otimes n}).
\end{align*}
where $\sigma$ is an arbitrary state on the output system of a single channel.
We now choose $\sigma$ as the optimal state in the Schumacher-Westmoreland
characterization of $\chi(\mathcal{N})$ \cite{SW01}:%
\[
\chi(\mathcal{N})=\min_{\sigma}\max_{\rho}D((\mathcal{N}(\rho)\Vert\sigma).
\]
For this, note also the previously used fact
\[
\widetilde{K}_{\alpha}^{[\sigma]}(\mathcal{N})\leq\chi(\mathcal{N}%
)+4(\alpha-1)(\log\nu)^{2}.
\]
Thus, we find the following bound on the success probability:%
\begin{equation}
p_{\text{succ}}=1-\varepsilon\leq2^{-n\left(  \frac{\alpha-1}{\alpha}\right)
(R-\frac{1}{n}\widetilde{K}_{\alpha}^{[\sigma^{\otimes n}]}(\mathcal{N}%
^{\otimes n}))}. \label{eq:H-error-bound}%
\end{equation}
The crucial observation, which in fact we also used to prove the strong
converse for entanglement-breaking channels, is that
\[
\widetilde{K}_{\alpha}^{[\sigma]}(\mathcal{N})=\max_{\rho}\frac{1}{\alpha
-1}\log\operatorname{Tr}\left\{  \left(  \sigma^{\frac{1-\alpha}{2\alpha}%
}\mathcal{N}(\rho)\sigma^{\frac{1-\alpha}{2\alpha}}\right)  ^{\alpha}
\right\}  ,
\]
which is $\frac{\alpha}{\alpha-1}$ times the logarithm of the maximum output
$\alpha$-norm of the sandwiched map%
\[
\left(  \mathcal{X}\circ\mathcal{N}\right)  (\rho)\equiv X\mathcal{N}(\rho)X,
\]
with $X=\sigma^{\frac{1-\alpha}{2\alpha}}$.

Now, we first prove that the strong converse holds for a Hadamard channel
$\mathcal{N}_{H}$ whose complementary channel $\mathcal{N}_{H}^{c}$ is in the
interior of the set of entanglement breaking channels.\footnote{Such channels
have the property that their Choi matrix is in the interior of the set of
separable states. That the interior of the set of entanglement-breaking
channels is non-empty then follows from \cite{PhysRevA.66.062311}.} In such a
case, $X=\sigma^{\frac{1-\alpha}{2\alpha}}$ becomes arbitrarily close to the
identity operator $I$ for $\alpha$ sufficiently close to one. (Without loss of
generality, we can assume that $\sigma$ has full rank---otherwise either
$\widetilde{K}_{\alpha}^{[\sigma]}(\mathcal{N}_{H})=+\infty$, or we can reduce
the size of the output system without affecting the performance of a given
code.) But then, the complementary map $\left(  \mathcal{X}\circ
\mathcal{N}_{H}\right)  ^{c}$ is arbitrarily close to $\mathcal{N}_{H}^{c}$,
and hence (always for sufficiently small $\alpha>1$) it is arbitrarily close
to a completely positive entanglement-breaking map. So it follows that
$\left(  \mathcal{X}\circ\mathcal{N}_{H}\right)  $ is a Hadamard map for
$\alpha$ sufficiently close to one, and
Theorem~\ref{thm:mult-from-complements}\ implies that its maximum output
$\alpha$-norm is multiplicative, so that the $\alpha$-information radius
around $\sigma$ is subadditive:%
\[
\frac{1}{n}\widetilde{K}_{\alpha}^{[\sigma^{\otimes n}]}(\mathcal{N}%
_{H}^{\otimes n})\leq\widetilde{K}_{\alpha}^{[\sigma]}(\mathcal{N}_{H}).
\]
Hence, from (\ref{eq:H-error-bound}), we find the following upper bound on the
success probability:%
\[
1-\varepsilon\leq2^{-n\left(  \frac{\alpha-1}{\alpha}\right)  (R-\widetilde
{K}_{\alpha}^{[\sigma]}(\mathcal{N}_{H}))}.
\]
By following the same steps as in Section~\ref{sec:final-steps-EB} (always
choosing $\alpha$ sufficiently close to one), the strong converse follows,
with a bound on the success probability that converges exponentially fast to zero.

For a Hadamard channel $\mathcal{N}_{H}$ whose complement $\mathcal{N}_{H}%
^{c}$ is on the boundary of the set of entanglement-breaking channels, the
argument above does not apply, since the perturbation inflicted by sandwiching
with $X\approx I$ might take the complementary channel outside the set of
entanglement-breaking maps. However, we can use the following continuity
argument: For $p\geq0$, consider the depolarizing channel on the environment
system $E$:%
\[
\mathcal{D}_{p}(\rho)=(1-p)\rho+p\frac{I}{|E|}\operatorname{Tr}\rho,
\]
with a suitable Stinespring isometry $W_{p}:E\rightarrow E\otimes F$, where
$|F|=|E|^{2}$. Then, not only is $\mathcal{M}_{p}^{c}\equiv\mathcal{D}%
_{p}\circ\mathcal{N}_{H}^{c}$ entanglement-breaking, but it is in the interior
of the set of entanglement-breaking channels whenever $p>0$. Furthermore, in
the limit as $p\rightarrow0$, $\mathcal{M}_{p}^{c}$ converges to
$\mathcal{M}_{0}^{c}=\mathcal{N}_{H}^{c}$. Hence, a similar limiting argument
applies for the map $\mathcal{M}_{p}$:%
\[
\mathcal{M}_{p}\rightarrow\mathcal{M}_{0}=\mathcal{N}_{H}\otimes
|0\rangle\!\langle0|,
\]
where $\mathcal{M}$ maps $A$ to $B\otimes F$, via $\mathcal{M}(\rho
)=\operatorname{Tr}_{E}\{W_{p}V\rho V^{\dagger}W_{p}^{\dagger}\}$. By the
continuity of the Holevo information $\chi$ in the channel \cite{LS09}, we
observe that $\chi(\mathcal{M}_{p})\rightarrow\chi(\mathcal{N}_{H})$.

Furthermore, $\mathcal{N}_{H}=\operatorname{Tr}_{F}\circ\mathcal{M}_{p}$, so
that every code for $\mathcal{N}_{H}$ is immediately a code with the same rate
and error parameters for $\mathcal{M}_{p}$. Now we can choose, for an
$n$-block code of rate $R>\chi(\mathcal{N}_{H})$ and error $\varepsilon$, a
$p>0$ such that $R>\chi(\mathcal{M}_{p})$. At this point the strong converse
follows for $\mathcal{M}_{p}$ by the previous argument, and hence also for
$\mathcal{N}_{H}$.

\section{Conclusion}

We have proven a strong converse theorem for the classical capacity of all
entanglement-breaking and Hadamard channels, and these results strengthen the
interpretation of the classical capacity for these channels. Our result
follows by obtaining tighter bounds on the success probability in terms of a
\textquotedblleft sandwiched\textquotedblright\ R\'{e}nyi relative entropy.
This information measure should find other applications in quantum information
theory, given that many other information measures can be obtained from a
relative entropy.


We have left the superadditivity of $\widetilde{\chi}_{\alpha}(\mathcal{N}%
_{1}\otimes\mathcal{N}_{2})$ for two channels as an open question, but Beigi
has recently provided a solution to this problem \cite{B13monotone}. That is,
Beigi has proved that the following inequality holds for any two channels:
\[
\widetilde{\chi}_{\alpha}(\mathcal{N}_{1}\otimes\mathcal{N}_{2})\geq
\widetilde{\chi}_{\alpha}(\mathcal{N}_{1})+\widetilde{\chi}_{\alpha
}(\mathcal{N}_{2}).
\]
Such an inequality for $\chi_{\alpha}$ easily follows---one can employ the
Sibson identity to find an explicit form for $\chi_{\alpha}$ and then the
inequality follows by simply choosing a suboptimal tensor product ensemble for
${\chi}_{\alpha}(\mathcal{N}_{1}\otimes\mathcal{N}_{2})$ (see Ref.~\cite{KW09}%
). However, it is not clear to us that a Sibson identity holds for
$\widetilde{D}_{\alpha}(\rho\|\sigma)$ except for when the states $\rho$ and
$\sigma$ are commuting. So the proof of the above inequality is more advanced
than the usual approach.

Finally, it might be possible to use the tools developed in
Refs.~\cite{MW12,M12} in order to prove strong converse theorems, but this
remains an open question.

\textit{Note}:\ After completing the work for the first version of this paper,
we discovered that other authors had already defined \cite{T12tutorial,F13}%
\ and proved \cite{L13,DST13}\ some of the properties of the sandwiched
R\'{e}nyi relative entropy. However, only the definition of the sandwiched
R\'{e}nyi relative entropy was publicly available at the time when we
completed this work. These authors have posted details of their work, now
published in Ref.~\cite{MDSFT13}.

Since our original arXiv post, there has been more activity in developing the
sandwiched R\'enyi relative entropy. In particular, M\"uller-Lennert
\textit{et al.}~have been able to prove many of their conjectures concerning
this quantity in a second version of their paper, while Frank and Lieb have
proved that it is monotone under quantum operations for all $\alpha\in[1/2,
\infty]$ \cite{FL13}. Simultaneously, Beigi provided a different proof that it
is monotone for all $\alpha\in(1,\infty)$ \cite{B13monotone}.

\textit{Acknowledgements}---We are grateful to Min-Hsiu Hsieh, Joe Renes, and
Graeme Smith for helpful discussions and to Marco Tomamichel for carefully
reading our paper and pointing out a correction to a previous version of
Proposition~\ref{prop:converge-to-vN}. We thank Fr\'{e}d\'{e}ric Dupuis, Marco
Tomamichel, and Serge Fehr for passing along
Refs.~\cite{T12tutorial,F13,L13,DST13}. MMW\ is grateful to the quantum
information theory group at the Universitat Aut\`{o}noma de Barcelona for
hosting him for a research visit during April-May 2013. AW's work is supported
by the European Commission (STREP \textquotedblleft QCS\textquotedblright),
the European Research Council (Advanced Grant \textquotedblleft
IRQUAT\textquotedblright) and the Philip Leverhulme Trust. DY's work is
supported by the ERC (Advanced Grant \textquotedblleft
IRQUAT\textquotedblright) and the NSFC (Grant No. 11375165).

\appendix

\section{Appendix}

We reproduce here, for convenience of the reader, the statements of Theorem
5.14, Corollary 5.5, and Theorem 5.16 from \cite{Wolf12}.

\begin{theorem}
[Theorem 5.16 \cite{Wolf12}]Let $g:\mathcal{D}\subseteq\mathcal{M}_{d_{1}%
}\times\cdots\times\mathcal{M}_{d_{n}}\rightarrow\mathcal{M}_{d}$ be a map on
the direct product $\mathcal{D}$ of $n$ positive operators, and similarly
$h:\mathcal{D}^{\prime}\subseteq\mathcal{M}_{d_{1}^{\prime}}\times\cdots
\times\mathcal{M}_{d_{n}^{\prime}}\rightarrow\mathcal{M}_{d}$. Suppose that
$g$ is jointly operator concave and positive and $h$ is semi-definite. Let
$I\ni0$ be the positive/negative real half line depending on whether $h$ is
positive or negative semi-definite. For any function $f:I\rightarrow
\mathbb{R}$ with $f\left(  0\right)  \leq0$, define $F:\mathcal{D}^{\prime
}\times\mathcal{D\rightarrow M}_{d}$ as%
\[
F\left(  L,R\right)  \equiv\sqrt{g\left(  R\right)  }f\left(  g\left(
R\right)  ^{-1/2}h\left(  L\right)  g\left(  R\right)  ^{-1/2}\right)
\sqrt{g\left(  R\right)  }.
\]
We consider joint operator convexity of $F$ in its $n+m$ arguments. $F$ is
jointly operator convex on positive operators for which $g$ is invertible if
at least one of the following holds: 1) $h$ is jointly operator concave and
$f$ is operator anti-monotone. 2)\ $h$ is affine and $f$ is operator convex.
\end{theorem}

\begin{corollary}
[Corollary 5.5 \cite{Wolf12}]$\mathcal{M}_{d}\times\mathcal{M}_{d}\ni\left(
L,R\right)  \rightarrow L^{x}\otimes R^{y}$ is jointly operator concave on
positive operators for $x,y\geq0$ with $x+y\leq1$.
\end{corollary}

\begin{theorem}
[Theorem 5.16 \cite{Wolf12}]Consider a functional $F:\mathcal{D}%
\subseteq\mathcal{M}_{d}\times\cdots\times\mathcal{M}_{d}\rightarrow
\mathbb{R}$ which is defined for all dimensions $d\in\mathbb{N}$. Suppose that
$F$ satisfies 1)\ joint convexity in $\mathcal{D}$, 2)\ unitary invariance,
i.e., for all $A\in\mathcal{D}$ and all unitaries $U\in\mathcal{M}_{d}\left(
\mathbb{C}\right)  $, it holds that $F\left(  UA_{1}U^{\dag},\ldots
,UA_{n}U^{\dag}\right)  =F\left(  A_{1},\ldots,A_{n}\right)  $, and
3)\ invariance under tensor products, meaning that for all $A\in\mathcal{D}$
and all density operators $\tau\in\mathcal{M}_{d^{\prime}}\left(
\mathbb{C}\right)  $, we have $F\left(  A_{1}\otimes\tau,\ldots,A_{n}%
\otimes\tau\right)  =F\left(  A_{1},\ldots,A_{n}\right)  $. Then $F$ is
monotone with respect to all CPTP\ maps $T:\mathcal{M}_{d}\left(
\mathbb{C}\right)  \rightarrow\mathcal{M}_{d^{\prime\prime}}\left(
\mathbb{C}\right)  $, in the sense that for all $A\in\mathcal{D}$,%
\[
F\left(  T\left(  A_{1}\right)  ,\ldots,T\left(  A_{n}\right)  \right)  \leq
F\left(  A_{1},\ldots,A_{n}\right)  .
\]

\end{theorem}

\bibliographystyle{plain}
\bibliography{Ref}

\end{document}